\documentclass[12pt]{article}
\usepackage{float}
\usepackage{placeins}
\usepackage{graphicx}
\usepackage{comment}
\graphicspath{ {figs/} }
\usepackage[english]{babel}
\usepackage{dsfont}
\usepackage[OT1]{fontenc}
\usepackage[utf8]{inputenc}
\usepackage{subcaption}
\usepackage{refcount}
\usepackage{lmodern}
\usepackage{booktabs}
\usepackage[dvipsnames]{xcolor} 
\usepackage{tikz}
\usepackage[normalem]{ulem}
\usepackage[symbol]{footmisc}

\let\href\undefined
\usepackage[colorlinks,citecolor=blue,urlcolor=blue,hyperfootnotes=false]{hyperref}

\usepackage{nameref,zref-xr}
\zxrsetup{toltxlabel} 
\usepackage{url}

\usepackage{mathtools}
\usepackage{amsfonts}
\usepackage{amssymb}
\usepackage{bbm}
\usepackage{amsthm}
\usepackage{amsxtra}
\usepackage{amsmath}
\usepackage{commath}
\usepackage{times}
\usepackage{bm}
\usepackage[linesnumbered,ruled,vlined]{algorithm2e}
\usepackage{cancel}

\usepackage[nameinlink]{cleveref}

\makeatletter%
\@ifclassloaded{biostatistics}{%
  \let\tiny\normalsize
  \let\scriptsize\normalsize
  \let\footnotesize\normalsize
  \let\small\normalsize
  \let\large\normalsize
  \let\Large\normalsize
  \let\LARGE\normalsize
  \let\huge\normalsize
  \let\Huge\normalsize
  \usepackage[inline,shortlabels]{enumitem}

  \usepackage[margin=1in]{geometry}

  \bibliographystyle{biorefs}

  \newcommand{\authorlist}{
    PHILIPPE BOILEAU$^\ast$, NIMA S.~HEJAZI$^\ast$, IVANA MALENICA$^\ast$,
    PETER B.~GILBERT,
    SANDRINE DUDOIT, MARK J.~VAN DER LAAN
    \\[4pt]
    \textit{
       Division of Biostatistics, School of Public
       Health, and Department of Statistics, University of California,
       Berkeley, 2121 Berkeley Way, Berkeley, CA 94720
    }
    \\[2pt]
    \textit{
       Department of Biostatistics, Harvard T.H.~Chan School of Public Health,
       677 Huntington Ave., Boston, MA 02115
    }
    \\[2pt]
     \textit{
       Vaccine \& Infectious Disease Division and Public Health Sciences
       Division, Fred Hutchinson Cancer Center,
       1100 Fairview Ave.~N., Seattle, WA 98109
     }
     \\[2pt]
    {\texttt{
      philippe\_boileau@berkeley.edu,
      nhejazi@hsph.harvard.edu,
      imalenica@berkeley.edu
    }}
  }

}{%
  \usepackage{arxiv}
  \usepackage{standalone}
  \usepackage{multirow}
  \usepackage{setspace}
  \usepackage{amsthm}
  \usepackage[round]{natbib}
  \usepackage[inline,shortlabels]{enumitem}
  \bibliographystyle{biometrika}

  \newtheorem{theorem}{Theorem}
  \newtheorem{corollary}{Corollary}
  \newtheorem{lemma}{Lemma}

  {\theoremstyle{definition}}
  {\theoremstyle{definition}}

  \newcommand{\authorlist}{
    Philippe Boileau\textsuperscript{\textdagger} \\
    Department of Epidemiology,\\
    Biostatistics  and Occupational Health,\\
    McGill University, and \\
    Research Institute of the McGill \\
    University Health Centre \\
    \texttt{philippe.boileau@mcgill.ca}
    \And
    Nima S.~Hejazi\textsuperscript{\textdagger} \\
    Department of Biostatistics, \\
    Harvard T.H.~Chan School of Public Health \\
    \texttt{nhejazi@hsph.harvard.edu}
    \And
    Ivana Malenica\textsuperscript{\textdagger} \\
    Department of Biostatistics, \\
    University of North Carolina at Chapel Hill \\
    \texttt{imalenic@unc.edu}
    \And
    Peter B.~Gilbert \\
    Vaccine \& Infectious Disease Division and \\
    Public Health Sciences Division, \\
    Fred Hutchinson Cancer Center \\
    \texttt{pgilbert@scharp.org}
    \And
    Sandrine Dudoit \\
    Department of Statistics and \\
    Division of Biostatistics, \\
    University of California, Berkeley \\
    \texttt{sandrine@stat.berkeley.edu}
    \And
    Mark J.~{van der Laan} \\
    Division of Biostatistics and \\
    Department of Statistics, \\
    University of California, Berkeley \\
    \texttt{laan@berkeley.edu}
  }
}

\newenvironment{assumptionidenalt}[1]{%
  \renewcommand{\@currentlabel}{\protect\ref*{#1}a}
  \renewcommand{\refstepcounter}[1]{}%
  \phantomsection
  \begin{assumptioniden}
}{%
  \end{assumptioniden}
}
\makeatother

\usetikzlibrary{shapes,decorations,arrows,calc,arrows.meta,fit,positioning}
\tikzset{
    -Latex,auto,node distance =1 cm and 1 cm,semithick,
    state/.style ={ellipse, draw, minimum width = 0.7 cm},
    point/.style = {circle, draw, inner sep=0.04cm,fill,node contents={}},
    bidirected/.style={Latex-Latex,dashed},
    el/.style = {inner sep=2pt, align=left, sloped}
}

\newcommand{\titlepaper}{Identifying direct causal effects under unmeasured
  exposure-mediator confounding}
\newcommand{\cofirst}{
  \textdagger~Co-first authors contributing equally (alphabetical ordering),
  to whom correspondence should be jointly addressed.
}

\AtEndDocument{\refstepcounter{theorem}\label{finalthm}}
\AtEndDocument{\refstepcounter{equation}\label{finaleq}}
\AtEndDocument{\refstepcounter{lemma}\label{finallemma}}

{\theoremstyle{definition} \newtheorem{assumptioniden}{}}

{\theoremstyle{definition} \newtheorem{assumptionb}{}}

\newcommand{\E}{\mathbb{E}}

\newcommand{\M}{\mathcal{M}}
\newcommand{\R}{\mathbb{R}}

\renewcommand{\P}{\mathbb{P}}
\newcommand{\Pt}{\mathsf{P}}
\newcommand{\I}{\mathbbm{I}}

\newcommand{\indep}{\mbox{$\perp\!\!\!\perp$}}
\DeclareMathOperator*{\argmin}{\arg\!\min}

\pdfoutput=1
\zexternaldocument*[supp:]{sm}

\makeatletter
\@ifclassloaded{biostatistics}{%
  \history{Received July 20, 2024?;
  revised N/A;
  accepted for publication N/A}

  \authorlist

  \markboth%
  {P.~BOILEAU, N.S.~Hejazi, I.~MALENICA et al.}
  {Identifying and Estimating Causal Direct Effects Under Unmeasured
  Confounding}
}{%
  \author{\authorlist}
}
\makeatother

\title{\titlepaper}
\begin{document}
\maketitle
\footnotetext{\cofirst}

\begin{abstract}
  Causal mediation analysis provides techniques for defining and estimating
  effects that may be endowed with mechanistic interpretations. With many
  scientific investigations seeking to address mechanistic questions, causal
  direct and indirect effects have garnered much attention.
  The natural direct and indirect effects, the most widely used among such
  causal mediation estimands, are limited in their practical utility due to
  stringent identification requirements. Accordingly, considerable effort has
  been invested in developing alternative direct and indirect effect
  decompositions with relaxed identification requirements. Such efforts often
  yield effect definitions with nuanced and challenging interpretations. By
  contrast, relatively limited attention has been paid to relaxing the
  identification assumptions of the natural direct and indirect effects.
  Motivated by a secondary aim of a recent non-randomized vaccine prospective
  cohort study (NCT05168813), we present a set of relaxed conditions under
  which the natural direct effect is identifiable in spite of unobserved
  baseline confounding of the exposure--mediator pathway; we use this result to
  investigate the effect mediated by putative immune correlates of protection.
  Relaxing the commonly used but restrictive cross-world counterfactual
  independence assumption, we discuss strategies for evaluating the natural
  direct effect in non-randomized settings that arise in the analysis of
  vaccine studies. We revisit prior studies of semi-parametric efficiency
  theory to demonstrate the construction of flexible, multiply robust
  estimators of the natural direct effect and discuss efficient estimation
  strategies that do not place restrictive modeling assumptions on nuisance
  functions. The proposed analytic approach is validated by numerical
  experiments and in an analysis of the motivating study with the aim of
  characterizing differences in activation of a well-studied immune correlate
  of protection for COVID-19 between non-randomized hybrid and vaccine immunity
  groups.
\end{abstract}

\section{Introduction}\label{intro}

While causal inference has placed a significant focus on delineating the total
effects of exposures, many causal effect definitions can fall short of providing
answers to mechanistic inquiries. When a scientific question concerns the effect
of an exposure on an outcome only through specific pathways between the two,
approaches investigating path-specific effects, through post-treatment variables
(known as ``mediators''), are of most interest; statistical approaches to
answering such questions fall under the umbrella of causal mediation analysis.
Recent examples include elucidating the biological mechanism by which a vaccine
reduces risk of infection or disease~\citep[for
example,][]{hejazi2020efficient, gilbert2022controlled, huang2023stochastic,
hejazi2023stochastic, gilbert2024four} or ascertaining the effects of
pharmacological therapies on substance abuse disorder relapse~\citep[for
example,][]{rudolph2021explaining, hejazi2022nonparametric,10.1093/aje/kwaf093}. One approach to
studying causal mediation analysis formalizes the partitioning of a total
effect into an \textit{indirect} effect, which captures the effect of the
exposure through pathways involving mediators identified \textit{a priori}, and
a \textit{direct} effect, which captures the remainder of the total effect, as
well as the development of appropriate techniques for robust statistical
inference about these mechanistically informative causal estimands.

Causal mediation analysis has a rich history, with the path analysis approach
of~\citet{wright1934method} perhaps being the earliest example. Subsequent
notable developments included parametric structural equation models~\citep[for
example,][]{goldberger1972structural, baron1986moderator} for statistical
mediation analysis, which garnered attention in the psychological and social
sciences. With the emergence of modern causal inference frameworks, including
the causally interpretable structured tree graph of~\citet{robins1986new} and
nonparametric structural equation models (NPSEMs)~\citep{pearl2000causality},
the necessary foundational tools to express causal mediation mechanisms without
reliance on restrictive parametric forms became available. As these and related
formalisms~\citep[for example,][]{spirtes2000causation} took root, nonparametric
formulations of mediation analysis uncovered significant limitations of earlier
efforts tied to parametric structural equation models~\citep{imai2010general},
demonstrating that the generally restrictive assumptions underlying such
approaches make them unsuitable for the study of the complex phenomena most
often of interest in the biomedical and health
sciences~\citep{nguyen2021clarifying}.

These modern developments in causal inference have yielded significant advances
over traditional mediation analysis techniques. A canonical example arose from
the works of~\citet{robins1992identifiability} and~\citet{pearl2001direct}, who used distinct frameworks to independently derive identification conditions for
path-specific effects defined by a nonparametric decomposition of the average treatment effect: the \textit{natural} (pure) direct and (total) indirect
effects (henceforth, ``(in)direct effects''). Identification of the natural
(in)direct effects requires several assumptions beyond those required for the corresponding total effect, though these assumptions---which include
unconfoundedness of the exposure--outcome, exposure--mediator, and
mediator--outcome relationships as well as a form of counterfactual
independence---are now considered standard in causal mediation analysis~\citep{vanderweele2015}. Despite the ubiquity of these assumptions for identification of these (in)direct effects, much effort has been aimed at
weakening these requirements. For example, the ``cross-world'' counterfactual independence assumption~\citep{andrews2020insights}, which requires independence of potential outcomes indexed by distinct interventions on exposure and
mediators, is incompatible with randomization~\citep{didelez2006direct,
robins2010alternative}, sharply limiting the applicability of these (in)direct
effect estimands. Put simply, by making corresponding scientific claims
impervious to falsification
by experimentation~\citep{popper1934logic}, this independence
assumption contradicts a central dictum of scientific inquiry. Such
counterfactual independencies are unsatisfied in the presence of intermediate
confounders affected by exposure as well~\citep{avin2005identifiability,
tchetgen2014identification}, which spurred study of the interventional
(in)direct effects~\citep[for example,][]{vanderweele2014effect,
vansteelandt2017interventional,diaz2020nonparametric,hejazi2022nonparametric}.
Recently, \citet{miles2023causal} showed that these (in)direct effect
definitions may, without stronger assumptions, fail to satisfy a mediational
sharp null criterion requiring that an indirect effect measure be null whenever
no individual-level indirect effect exists; thus, in practice, such
interventional effect measures carry the risk of misleading as to whether a
mechanistic indirect effect truly exists.

Another set of unique and critical assumptions is that of no unmeasured
confounding of the exposure--mediator and mediator--outcome
mechanisms~\citep[for example,][]{imai2010identification}. When considering
total effects, randomization renders such an assumption unnecessary by breaking
any putative causal link between baseline confounders and exposure, thereby
allowing for the effect of the exposure on the outcome to be disentangled from
those of any confounding factors. Put another way, randomization enforces
unconfoundedness of the exposure--outcome and exposure--mediator relationships
by design. The incompatibility of the natural (in)direct effects'
identification conditions with randomization exacerbates their reliance on
these assumptions. Given that such confounders are challenging to rule out in
practice, the practical utility of the natural (in)direct effects may be
limited. Accordingly, attention---though still limited---has been focused on
weakening these requirements. Recent examples of note include the work
of~\citet{vansteelandt2012natural}, who formulated novel (in)direct effects
arising from a decomposition of the average treatment effect on the treated;
that of~\citet{fulcher2019robust}, who provided a less stringent set of
assumptions for identification of an indirect effect; and that
of~\citet{rudolph2021complier}, who developed a robust direct effect among
compliers in a setting with instrumental variables. The (in)direct effects
of~\citet{vansteelandt2012natural} may be learned even in the presence of
unmeasured confounding of the exposure--mediator relationship while
\citet{fulcher2019robust}'s indirect effect remains identified under unmeasured
confounding of the exposure--outcome pathway, and~\citet{rudolph2021complier}'s
complier direct effect relies on randomization of an instrument to remain
robust to both observed and unobserved confounding of the exposure--outcome
relationship. 
Recently, \citet{Stensrud2024} studied identification and estimation of causal
effects of intervening variables (i.e., mediators) in settings where standard
assumptions fail, in particular due to unmeasured exposure--outcome confounding;
their approach, similar to that of~\citet{fulcher2019robust}, replaces
sequential ignorability with an identification strategy that leverages
intervening variables, closely related to the front-door
criterion~\citep{pearl2000causality}, which partitions a target effect along
pathways involving intervening variables.


These advances are particularly relevant for comparative vaccine studies of outcome risk,
which, as highlighted by the COVID-19 pandemic, is of critical importance to
public health. In vaccine research, a central goal is the characterization of
immune correlates of protection, that is, vaccine-induced immune markers that help to
elucidate immunological mechanisms and may be validated to serve as surrogate
endpoints for a relevant clinical outcome, thereby reducing reliance upon
resource-intensive randomized controlled trials~\citep{gilbert2024four}. Causal
mediation analysis permits the decomposition of vaccinations' protective effects
with respect to disease endpoints into direct and indirect effects, the latter
carried through candidate immune correlates of protection. Defining and
evaluating causal effects along pathways involving putative immune correlates
facilitates the mechanistic assessment of vaccine efficacy; moreover, scientific
insights grounded in such statistical evidence can inform the development of
immunologic targets for future vaccines, which may subsequently be interrogated
in clinical trials to ultimately advance the state of global public health.
Inference about these direct and indirect effects is often made difficult,
however, by unmeasured confounding. Sources of unmeasured confounding are
varied, including genetic factors predisposing individuals to heightened or
suppressed immune responses, unmeasurable-at-baseline immunocompromised status, prior infection, or, indeed, future infection risk,  to name but a few.

Our developments are motivated by their application to addressing a secondary
aim of the COVID-19 Prevention Network's (CoVPN) 3008 study (Ubuntu;
NCT05168813), a prospective study of Moderna's mRNA-1273 COVID-19 vaccine to evaluate the effect of two immunity conditions in preventing COVID-19
disease in people living with HIV or with comorbidities indicative of an
elevated risk of severe COVID-19 disease. In CoVPN 3008, study participants
received either one or two doses of the study vaccine based on evidence of
prior SARS-CoV-2 infection detectable by serology at
enrollment~\citep{garrett2025hybrid}. A critical secondary objective of this
study is to evaluate the activity of previously studied immune correlates of
protection for COVID-19 disease, such as neutralizing antibody titer, between
non-randomized groups that either received two doses of the mRNA-1273 vaccine
and lacked evidence of prior infection (vaccine immunity) or received one dose
of the mRNA-1273 vaccine but displayed evidence of prior infection (hybrid
immunity). Such a comparison informs as to how putative immune correlates of
protection may be related to the choice of vaccine regimen (as informed by
prior infection history and comorbidities) and COVID-19 disease outcomes
post-vaccination, thereby allowing for validated immune correlates of
protection to inform guidance on the basis of vaccine or hybrid immunity. For
such findings to be policy relevant, the characterization of immune correlates
of protection must be robust to the presence of unmeasured confounding of the
exposure--mediator relationship---for example, by neutralizing antibody titer
at baseline (prior to vaccination) or baseline immunoglobulin G (IgG) antibody
concentration against the viral N (nucleocapsid) protein that is not represented in the vaccine under study (and hence is informative of infection history). Even when such confounders
are documented, the use of such measurements as part of vaccine guidance is
neither economical nor practical.

To this end, we propose a novel identification result that permits inference
about the natural direct effect
in the presence of unmeasured confounding of the exposure--mediator
relationship that is of particular relevance to comparative vaccine studies. We outline in Section~\ref{sec:stat-prob}
the underlying complete
data and observed data models, and formally define our causal estimand---the
natural direct effect. We then provide our identification result in
Section~\ref{sec:identification}, which
corresponds to a previously studied statistical parameter. A review of
inference procedures for the natural direct effect is provided in
Section~\ref{sec:inference}. This is followed by a
corroborative simulation studies of our theoretical claims, in
Section~\ref{sec:sim},
and their applicability under outcome-dependent sampling schemes often used in
prevention cohort studies~\citep{follmann2006augmented}, in
Section~\ref{sec:applic-prospect}.
Following this, in
Section~\ref{sec:applic-covpn3008},
the proposed strategy is used to evaluate the proportion of the hybrid versus vaccine immunity effect
mediated through a well-studied immune correlate of
protection~\citep{gilbert2022covid, gilbert2024four} on subsequent COVID-19
disease involving a comparison of non-randomized groups. Section~\ref{sec:discussion}
concludes with a brief discussion of the implications
of our contributions for causal mediation analysis in vaccine studies.


\section{Formulation of the Statistical Problem}\label{sec:stat-prob}

\subsection{The Complete (Unobserved) Data Model}\label{causal_model}

Translation of the scientific question of interest into a causal parameter can
be facilitated by an NPSEM, or, equivalently, a structural causal model (SCM).
In specifying an SCM, we assume that each component of the data structure
is a function of an appropriate subset of the endogenous variables $\{W, V, A,
Z, Y\}$ and exogenous variables $\{U_W, U_V, U_A, U_Z,
U_Y\}$~\citep{pearl2000causality}. In temporal order, the endogenous variables
constitute baseline covariates $W \in \mathcal{W}$, \textit{unmeasured} baseline
confounders $V \in \mathcal{V}$, exposure $A \in \mathcal{A}$, mediators $Z \in
\mathcal{Z}$, and outcome $Y \in \mathcal{Y}$. In our analysis of CoVPN 3008, which follows results reported by~\citet{garrett2025hybrid}, $W$ includes
a pre-specified adjustment set consisting of enrollment region, enrollment
period, baseline HIV status, baseline TB status, and a baseline infection risk
score; $V$ is taken to be IgG concentration against the viral N
protein, informative of infection history and recency but unmeasured in most
study participants; $A$ is receipt of one dose of mRNA-1273 vaccine and
evidence of prior infection ($A=1$, that is, hybrid immunity) versus receipt of
two doses of mRNA-1273 vaccine and no evidence of prior infection ($A=0$,
that is, vaccine immunity); $Z$ is neutralization titer against the BA.4/BA.5
variant of SARS-CoV-2 (ID50, $\log_{10}$-scaled) four weeks after receipt of
the last vaccine dose; and $Y$ is an indicator of COVID-19 disease based on the
U.S.~CDC definition between 7 and 165 days after the day 29 visit post final
vaccination. We collectively refer to endogenous variables as $X=(W,V,A,Z,Y)$,
which corresponds to the \textit{complete}, but unobservable, data on a single
study unit. The unmeasured exogenous variables, denoted as $U = (U_W, U_V, U_A,
U_Z, U_Y)$, are sampled from the unknown joint probability distribution
$\Pt_U$. The following set of structural equations illustrates dependencies
allowed by the assumed time-ordering between the variables in the system under
study:
\begin{align}\label{SEM}
    W &= f_W(U_W)  \\
    V &= f_V(U_V) \nonumber \\
    A &= f_A(W,V,U_A) \nonumber \\
    Z &= f_Z(W,V,A,U_Z) \nonumber \\
    Y &= f_Y(W,V,A,Z,U_Y), \nonumber
\end{align}
where $\mathcal{F}_X=\{f_W,f_V,f_A,f_Z,f_Y\}$ denote deterministic functions of
unrestricted functional form. In practice, knowledge of the data-generating
experiment may be used to inform the structure of some of these
quantities---for example, in a randomized controlled trial (RCT), the form of
$f_A$ would be known exactly. Generally, knowledge of the forms of the
structural equations $\{f_W, f_V, f_Z, f_Y \}$ is rarely available or severely
limited. For further visual clarity, we depict relationships between variables
corresponding to the SCM~\eqref{SEM} in
Figure~\ref{fig:DAG_full}.

The SCM defines a collection of distributions over $(U, X)$, which constitute
the complete data model; accordingly, we denote the true probability
distribution of the pair $(U, X)$ as $\Pt_{(U, X), 0}$. Throughout the
remainder of the text, the naught subscript is used to indicate the
\textit{true}, or data-generating, probability distributions, or components
thereof. We denote by $\mathcal{M}^F$ the complete data model, noting that
$\Pt_{(U,X),0} \in \mathcal{M}^F$.

\begin{figure}[H]
    \centering
    \begin{tikzpicture}
    \node[state] (1) at (0,3.5) {$W$};
    \node[state] (2) at (0,2) {$Z$};
    \node[state] (3) at (0,-0.5) {$V$};
    \node[state] (4) at (-2,1) {$A$};
    \node[state] (5) at (2,1) {$Y$};
    \node[state,rectangle] (6) at (-2.5,3) {$U$};

    \path (1) edge (2);
    \path (1) edge (5);
    \path (1) edge (4);

    \path (4) edge (5);
    \path (4) edge (2);
    \path (2) edge (5);

    \path (3) edge (4);
    \path (3) edge (2);
    \path (3) edge (5);

    \path (6) edge[densely dotted] (1);
    \path (6) edge[densely dotted] (2);
    \path (6) edge[densely dotted] (3);
    \path (6) edge[densely dotted] (4);
    \path (6) edge[densely dotted] (5);

\end{tikzpicture}
    \caption{Directed acyclic graph (DAG) corresponding to SCM~\eqref{SEM}.
    We depict the endogenous variables with circles and exogenous variables
    with rectangles. For the sake of visual economy, we write all exogenous
    variables under the common set of exogenous factors $U$.}
    \label{fig:DAG_full}
\end{figure}
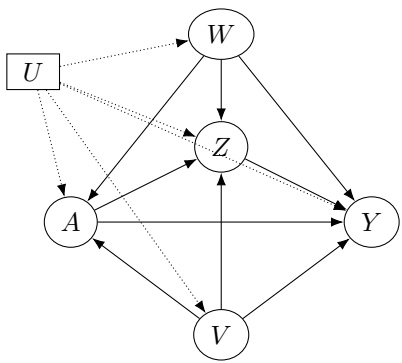

\subsection{The Observed Data Model}\label{statistical_model}

We consider access to cohort data consisting of independent and identically
distributed (i.i.d.) copies of the random variable $O = (W, A, Z, Y)$. We
emphasize that the random variable $O$ does not include unmeasured baseline
covariates $V$, though the complete data random variable $X$ does.
Let $O \sim \Pt_0 \in \M$, where $\Pt_0$ is the data-generating distribution of
$O$ and $\M$ is the \textit{observed data model}. We construct $\M$ without
constraining the form of $\Pt_0$---outside any knowledge possibly available on
the exposure mechanism, as may be the case in an RCT. As with the complete data
model, we may limit the observed data model in realistic ways by leveraging
scientific knowledge available about the system under study and the process
producing the data. We assume access to $n$ copies of $O$, comprising an
i.i.d.~sample $O_1, \ldots, O_n$, and we use $\Pt_n$ to denote the sample's
empirical distribution, which weights each observation equally. Estimates
derived from the empirical distribution are denoted with the subscript $n$
throughout. Given $\Pt_n$, we can estimate functionals of $\Pt_0$. Ideally,
then, these estimates also relate to the parameters of the complete data distribution,
$\Pt_{(U,X),0}$, endowing them with causal interpretations. This is impossible,
however, without further unverifiable identification conditions.

We now define various quantities of the observed data model which will be
subsequently used to define or prove these conditions. Let $g_A(W)$ denote the
conditional probability of exposure given baseline covariates, that is,
$g_A(W) \coloneqq \P(A = 1 \mid W)$, and denote the post-intervention
conditional density of the mediators as $g_Z(z \mid W, A)$. We also define
$\overline{Q}_Y(W, A, Z) \coloneqq \E[Y \mid W, A, Z]$ as the conditional
expectation of the outcome $Y$ given the variables that precede it in the temporal
ordering of the SCM. Finally, $p_{W}(w)$ represents the marginal density of $W$.
Nuisance parameters represented by ``$g$'' are those used to define potential
outcomes via hypothetical interventions.

\subsection{Static and Stochastic Interventions}\label{interventions}

The complete data model, $\M^F$, allows us to define counterfactual random
variables as functions of $(U, X)$ corresponding with the actions of specific
interventions on the system specified by the SCM~\eqref{SEM}.
These counterfactual quantities arise from restrictions on the inputs to
particular components of $\mathcal{F}_X$.
To take a common example, a \textit{static} intervention entails setting the
value of $A$ to a value $a$ in its support $\mathcal{A}$. Such an intervention
can be thought of as removing the equation$f_A$ from the SCM~\eqref{SEM} and
replacing its natural output with $a$ in all subsequent components of $f_X$.
This would then generate the counterfactual $Z(a)$, arising from $f_Z(W, V, a,
U_Z)$, and nested counterfactual $Y(a, Z(a))$, arising from $f_Y(W, V, a, Z(a),
U_Y)$. The counterfactual $Z(a)$ may be interpreted as the value that the
mediating variables ``naturally'' would take had the exposure been set to $a \in
\mathcal{A}$. Similarly, the nested counterfactual $Y(a, Z(a))$ may be
interpreted as the value the outcome would take had the exposure been set to $a$
and had the mediating variables taken the value $Z(a)$ arising ``naturally''
under the intervention that sets the exposure $A$ to value $a$ in its support.

Static interventions do not necessarily readily translate
into actionable policies. For example, the CoVPN biostatistics response team
used, among other techniques, a stochastic-interventional vaccine efficacy
measure~\citep{hejazi2023stochastic} derived from a modified treatment
policy~\citep{robins2004effects, haneuse2013estimation} to estimate
counterfactual infection risk~\citep{hejazi2020efficient}. Such an intervention
scheme does not require conceiving of setting individual immune responses to a
specific value for all study participants, instead respecting the fact that
vaccination alters a study participant's immune response relative to its
pre-vaccination natural value~\citep{sarvet2025natural}. This vaccine efficacy
measure has been used to evaluate putative immune correlates of protection in
several CoVPN vaccine efficacy trials, including the CoVPN COVE trial of the
mRNA-1273 vaccine~\citep{huang2023stochastic, hejazi2023stochastic,
gilbert2024four}.
A more flexible intervention scheme is given by the framework of
\textit{stochastic} interventions, in which the post-intervention conditional
distributions of intervention variables (for example, the exposure and/or
mediators) are replaced by user-defined candidate distributions~\citep[for
example,][]{stock1989nonparametric, didelez2006direct, diaz2012population,
diaz2020causal}. This general framework accommodates other intervention schemes
too. For example, static interventions may be viewed as a special case in which
the candidate post-intervention conditional distribution places all mass on the
intervention level of interest (for example, a degenerate, distribution at $a
\in \mathcal{A}$). In the preceding example, the stochastic intervention
corresponding to setting $A=a$ would simply replace $f_A$ with draws from a
distribution with all mass placed upon the exposure level of interest $a \in
\mathcal{A}$. The resulting estimands are associated with implausible
intervention strategies for select individuals, thereby providing limited
utility as targets of inference to support biomedical, health, and policy
decisions~\citep{haneuse2013estimation}.

\subsection{Defining the Natural (In)Direct Effects}\label{causal_sem}

Causal target parameters are defined as functionals of the complete data
distribution $\Pt_{(U, X), 0}$. In particular, we may express the total effect
of an intervention on the exposure, or the average treatment effect (ATE), as
\begin{equation}\label{eqn:ate}
  \Psi_{\text{ATE}}^F(\Pt_{(U,X),0}) = \E_{(U,X),0}[Y(1) - Y(0)] \ ,
\end{equation}
where $Y(0)$ and $Y(1)$ are \textit{potential
outcomes}~\citep{rubin1974estimating, rubin2005causal} of $Y$, that is, mutually
unobservable counterfactual quantities that arise when the exposure takes the
values $A = 0$ and $A = 1$, respectively. Using different
frameworks,~\citet{robins1992identifiability} and~\citet{pearl2001direct}
introduced a decomposition of the ATE into the natural direct and indirect
effects, which have been the subject of much study in causal mediation
analysis~\citep[for example,][]{avin2005identifiability,
didelez2006direct,tchetgen2014identification}. The natural (or pure) direct
effect (NDE) and natural (or total) indirect effect (NIE) are obtained via the
following canonical decomposition of the ATE:
\begin{align}\label{eqn:nat_decomp}
  \Psi_{\text{ATE}}^F(\Pt_{(U,X),0}) &= \E_{(U,X),0}[Y(1) - Y(0)] \nonumber \\
    &= \underbrace{\E_{(U,X),0}[Y(1, Z(1)) - Y(1, Z(0))]}_{\text{NIE}} +
       \underbrace{\E_{(U,X),0}[Y(1, Z(0)) - Y(0, Z(0))]}_{\text{NDE}} \ ,
\end{align}
where $Y(a) = Y(a, Z(a)) = f_Y(W, V, a, Z(a), U_Y)$ and $Z(a) = f_Z(W, V, a,
U_Z)$ are the counterfactuals of the outcome and mediators, respectively, under
the static intervention that sets $A=a$. In particular, $Y(a)$ is the
counterfactual outcome under exposure level $a$ and the value of $Z$ that arises
under that intervention, $Z(a)$. In Equation~\eqref{eqn:nat_decomp}, the
decomposition term $Y(1, Z(0))$ is used to partition the total effect into paths
involving the mediators and those avoiding the mediators. Note that the NIE
corresponds to the former, capturing the portion of the total effect that passes
through the mediators $Z$, while the NDE merely captures the remainder of the
total effect---through all pathways that do not involve $Z$. Importantly, this
means that the substantive interpretation of the NIE is \textit{usually} more
specific than that of the NDE.

Examination of the decomposition of the ATE into the NDE and NIE reveals
that the natural (in)direct effects are defined by a joint static intervention
that gives rise to the decomposition term $Y(1, Z(0))$. This
counterfactual arises from the static intervention $A = 1$ and the
static intervention that sets the mediators $Z$ to the value that would have
been observed had the exposure been set to the opposite condition, that is,
$Z(0) = f_Z(W, V, 0, U_Z)$. While technically convenient, such joint static
interventions are logically problematic, requiring an incompatibility between the
interventions applied to $A$ and $Z$.~\citet{robins1992identifiability} resolved
this issue by introducing a form of ``cross-world counterfactual independence,''
which has been the subject of much debate and discussion in the causal mediation
analysis literature~\citep[for example,][]{avin2005identifiability,
robins2010alternative, diaz2020causal, andrews2020insights}. An alternative
intervention scheme, described by~\citet{vdl2008direct}, instead employs the
familiar static intervention to the exposure but couples this with a stochastic
intervention on the mediators. In particular, the proposed intervention scheme
is characterized as jointly applying a static intervention on the exposure $A=a$
and drawing mediator values from $Z^{\star} \sim \P(Z < z \mid W, 0)$, the
cumulative distribution function constructed from $g_{Z}(z \mid W, 0)$. In this
intervention scheme, the counterfactual exposure value is deterministically set
while the counterfactual mediator value is drawn from the specified conditional
distribution. This joint intervention does not require setting $Z$ through
$f_{Z}(W, V, 0, U_Z)$, thus avoiding the logical incompatibility embedded in the
foundational developments of~\citet{robins1992identifiability}. Under this joint
intervention, the NDE may be defined as
\begin{equation}\label{eqn:nde_stoch}
  \Psi_{\text{NDE}}^F(\Pt_{(U, X), 0}) =
    \E_{(U,X),0}[Y(1, Z^{\star}) - Y(0, Z^{\star})] \ .
\end{equation}
This clarifies that the NDE may be viewed as the contrast of the counterfactual
mean outcomes between the static interventions deterministically setting the
exposure to $A = 1$ and $A = 0$ while simultaneously drawing the mediators $Z$
from a cumulative distribution function (CDF) constructed from $g_Z(z \mid W,
0)$. Of course, as previously noted, the static interventions applied to the
exposure $A$ may also be viewed as stochastic interventions in which the
post-intervention conditional distribution of the exposure is degenerate,
placing all mass on the exposed condition.

\subsection{Analogs of Natural (In)Direct Effects in Vaccine Efficacy Studies}

Convention regarding the statistical assessment of vaccine efficacy (VE) defines
efficiency parameters typically as $\psi_{\text{VE}} = 1 - \psi_{\text{RR}}$,
where $\psi_{\text{RR}}$ is a risk ratio (RR) contrasting the study arms' risks.
Adhering to this convention, we follow~\citet{gilbert2021immune}
and~\citet{benkeser2021inference}, who express the total effect of vaccination
in terms of this risk ratio and further decompose this quantity into analogs of
the natural (in)direct effects:
\begin{equation}\label{eqn:rr_nat_decomp}
  \psi_{\text{RR}} = \frac{\P(Y(1, Z(1)) = 1)}{\P(Y(0, Z(0)) = 1)} =
    \underbrace{\frac{\P(Y(1, Z(0)) = 1)}{\P(Y(0, Z(0)) = 1)}}_{\text{RR NDE}}
    \underbrace{\frac{\P(Y(1, Z(1)) = 1)}{\P(Y(1, Z(0)) = 1)}}_{\text{RR NIE}} \ .
\end{equation}
Unlike the canonical decomposition of the ATE into the NDE and NIE, as presented
in Equation~\eqref{eqn:nat_decomp}, the decomposition of the RR is performed on
a multiplicative scale. Note that, to recover an additive decomposition
analogous to that of Equation~\eqref{eqn:nat_decomp}, one need only consider
instead $\log\psi_{\text{RR}}$. Inspection of Equation~\eqref{eqn:rr_nat_decomp}
reveals that the familiar decomposition term $\P(Y(1, Z(0)) = 1)$ appears in the
numerator of the RR NDE and in the denominator of the RR NIE. Identification and
efficient estimation of this decomposition term are therefore central to
utilizing these (in)direct effects in developing a mechanistic assessment of
vaccine efficacy through candidate immune correlates of protection.

Importantly, both~\citet{gilbert2021immune} and~\citet{benkeser2021inference}
make use of standard identification assumptions to establish conditions under
which the RR NDE may be learned from observed data. This includes the
complete set of assumptions about no unmeasured confounding---that all potential
confounders of $A$, $Z$, and $Y$ must be measured. Our developments
significantly relax these requirements, allowing for the possibility of
unmeasured-at-baseline confounders $V$ that affect the $A$--$Z$ pathway. In
vaccine studies where the exposure of interest is non-randomized, as is the case
when contrasting hybrid and vaccine immunity in the CoVPN 3008 study, $V$ may
represent, for example, unknown or unmeasured-at-baseline prior infection
history. Injecting this additional degree of flexibility significantly extends
the scenarios in which these multiplicative mediational effects can be inferred,
in turn strengthening the evidence generated through comparative vaccine
studies' immune correlates analyses.

\section{Identification}\label{sec:identification}

\subsection{Standard Identification Assumptions}\label{subsec:ident:ss}

We stress that the complete data distribution, $\Pt_{(U, X), 0}$, determines a
corresponding distribution $\Pt_0$ within the observed data model $\M$. By
defining the complete data parameter of interest in terms of interventions on
the SCM, and thereby linking the complete data model and the observed data, we
may obtain inference about these causal effects using the observed data.
Standard conditions under which the natural (in)direct effects are identified
have been addressed extensively in prior work~\citep[for
example,][]{robins1992identifiability, petersen2006estimation, imai2010general,
robins2010alternative, pearl2011, fulcher2019}. Typically, identification of
the natural (in)direct effects relies on strong assumptions regarding the
absence of unmeasured confounding along several pathways. These assumptions
render causal mediation analysis difficult to implement, justify, and interpret,
as, altogether, they require no unmeasured confounding of the exposure--outcome,
exposure--mediator, and mediator--outcome pathways, as well as no
exposure-induced confounding of the mediator--outcome
pathway~\citep{fulcher2019, diaz2020nonparametric, hejazi2022nonparametric,
miles2023causal}. Next, we outline the typical assumptions necessary for
identification of the NDE and NIE:%
\begin{assumptioniden}[\textit{Exposure--outcome conditional randomization}]
\label{ass:original_A1}
  $Y(a) \indep A \mid W~\forall~a \in \mathcal{A}$.
\end{assumptioniden}
\begin{assumptioniden}[\textit{Exposure--mediator conditional randomization}]
\label{ass:original_A2}
  $Z(a) \indep A \mid W~\forall~a \in \mathcal{A}$.
\end{assumptioniden}
\begin{assumptioniden}[\textit{Mediator--outcome conditional randomization}]
  \label{ass:original_A3}
  $Y(a) \indep Z(a^{\star}) \mid (W, A=a)~\forall~(a,a^{\star}) \in
  \mathcal{A}$.
\end{assumptioniden}
\begin{assumptioniden}[\textit{Positivity for the exposure}]\label{ass:posA}
$\epsilon_A < g_A(W) < 1-\epsilon_A$ almost surely (a.s.) for some $\epsilon_A > 0$.
\end{assumptioniden}
\begin{assumptioniden}[\textit{Positivity for mediators}]\label{ass:posZ}
$\epsilon_Z < g_Z(z \mid W,A < 1-\epsilon_Z$  for some $\epsilon_Z > 0$ a.s. $\forall~z\in\mathcal{Z}$.
\end{assumptioniden}

Assumptions~\ref{ass:original_A1} and~\ref{ass:original_A2} state that there is
no unmeasured confounding of the $A$--$Y$ and $A$--$Z$ relationships,
respectively. These assumptions are typically guaranteed by design in RCTs (and
without conditioning on $W$), as the exposure mechanism is known and
randomization is controlled by the experimenter. However, in observational
studies, Assumptions~\ref{ass:original_A1} and~\ref{ass:original_A2} can prove
problematic, as they amount to strong assertions about unconfoundedness along
several pathways. Furthermore, even if the exposure mechanism were known, as in
an RCT, bias due to unmeasured confounding of the $Z$--$Y$ pathway remains.
Assumption~\ref{ass:original_A3} encodes a condition of no unmeasured
confounding of the effect of $Z$ on $Y$ (for example, by an exposure-induced,
intermediate confounder of the $Z$--$Y$ relationship), which is coupled with a
more stringent assumption of independence of counterfactual mediators
and outcomes under possibly conflicting values of the exposure. Taken together,
reliance on Assumptions~\ref{ass:original_A1},~\ref{ass:original_A2}
and~\ref{ass:original_A3} rules out existence of unmeasured variables that
confound the $A$--$Y$, $A$--$Z$, and $Z$--$Y$ relationships, which could be
untenable in practice, especially in observational studies.
Assumptions~\ref{ass:posA} and~\ref{ass:posZ} are standard in causal mediation
analysis, requiring that the conditional probability of $A$, given $W$, and the
conditional density of $Z$, given $(A, W)$, be well-defined, allowing for the
mediation formula to be tractable~\citep{pearl2001direct,
vanderweele2017mediation}. Furthermore, these assumptions are practically
required to ensure numerical stability of certain estimators.

We emphasize here that, under the stochastic intervention formulation of the
NDE, the cross-world assumptions in~\ref{ass:original_A2}
and~\ref{ass:original_A3} can be relaxed: only conditional independence of
variables based on the observed past (that is, sequential randomization) is
necessary to ensure identification. As such, Assumptions~\ref{ass:original_A1},
\ref{ass:original_A2a}, \ref{ass:original_A3a}, \ref{ass:posA}, and
\ref{ass:posZ} are needed for causal identification of this
parameter~\citep{vdl2008direct}, with Assumptions~\ref{ass:original_A2a}
and~\ref{ass:original_A3a} defined as%
\begin{assumptionidenalt}{ass:original_A2}[\textit{No exposure--mediator
  confounding}]\label{ass:original_A2a}
   $Z \indep A \mid W$.
\end{assumptionidenalt}
\begin{assumptionidenalt}{ass:original_A3}[\textit{No mediator--outcome
  confounding}]\label{ass:original_A3a}
  $Y(a,z) \indep Z \mid W~\forall~(a, z) \in \mathcal{A} \times \mathcal{Z}$.
\end{assumptionidenalt}

Next, we will weaken these assumptions to allow for unmeasured confounders of
the exposure--mediator relationship.

\subsection{Relaxed Identification Assumptions}

We now establish identification conditions for the NDE in the presence of
unmeasured confounding of the exposure--mediator relationship. In
Section~\ref{sec:ate-ident} of the Appendix, we extend this result to identify the ATE, under an additional assumption on the
distribution of $Z(a) \mid W$. We show that Assumptions~\ref{ass:original_A1}
and~\ref{ass:original_A2} can be modified in order to allow for unmeasured
confounders $V$ to have a direct effect on both the exposure and the mediators.
We emphasize that distinct components of $V$ can affect $A$ and $Z$. At
times, it will prove useful to consider a representation of the model with
assumptions stated in terms of the DAG in Figure~\ref{fig:DAG},
which depicts this setup in which the unmeasured set of confounders $V$ affects
both the exposure and mediators but not the outcome.
\begin{figure}[H]
    \centering
    \begin{tikzpicture}
    \node[state] (1) at (0,3.5) {$W$};
    \node[state] (2) at (0,2) {$Z$};
    \node[state,rectangle] (3) at (0,-0.5) {$V$};
    \node[state] (4) at (-2,1) {$A$};
    \node[state] (5) at (2,1) {$Y$};

    \path (1) edge (2);
    \path (1) edge (5);
    \path (1) edge (4);

    \path (4) edge[densely dotted] (5);
    \path (4) edge[dashed] (2);
    \path (2) edge[dashed] (5);

    \path (3) edge (4);
    \path (3) edge (2);

    \node[draw=red,loosely dotted,fit=(3) (5), inner sep=0.2cm] (machine) {};
\end{tikzpicture}
    \caption{DAG with unmeasured confounders $V$ affecting both exposure ($A$)
    and mediators ($Z$) but with no direct path between $V$ and the outcome
    ($Y$), reflecting Assumptions~\ref{ass:no-V-to-Y} and \ref{ass:equalE}. The
    direct effect of $A$ on $Y$ is denoted via a dotted line, and
    the indirect effect through $Z$ with a dashed line. The lack of direct path
    between $V$ and $Y$ is emphasized through the red dotted rectangle. We
    depict the observed variables with circles and unobserved variables with
    rectangles, omitting $U$ from the graph for the sake of visual economy.}
    \label{fig:DAG}
\end{figure}
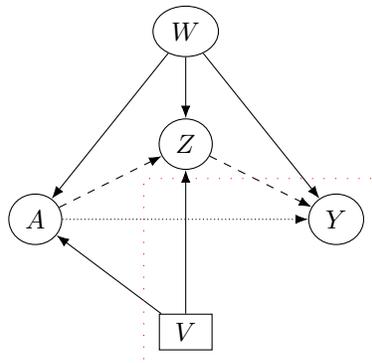

\begin{assumptioniden}[\textit{No unmeasured endogenous pathway}]
\begin{equation*}
\label{ass:no-V-to-Y}
    f_Y(W, V, A, Z, U_Y) \equiv f_Y(W, A, Z, U_Y) \ .
\end{equation*}
\end{assumptioniden}
\vspace{1mm}
\begin{assumptioniden}[\textit{Conditional expectation equivalence}]
\label{ass:equalE}
\begin{equation*}
    \E[Y \mid W,V,A=1,Z] = \E[Y \mid W,A=1,Z] \ .
\end{equation*}
\end{assumptioniden}

In line with the DAG of Figure~\ref{fig:DAG}, Assumption~\ref{ass:no-V-to-Y}
requires $f_Y$ to be a function of $V$ only through the effect $V$ has on
$(A,Z)$. Assumption~\ref{ass:equalE} implies the equivalence of the conditional
expectation of $Y$ with and without conditioning on the unmeasured confounders
$V$. While we use Assumption~\ref{ass:equalE} in the consequent identification
result for the NDE, we elaborate on the additional restrictions on the
distribution of $U$ implied by Assumption~\ref{ass:equalE} in
Section~\ref{sec:equivE} of the Appendix. In
brief, Assumption~\ref{ass:equalE} could imply restrictions on the joint
distribution of $U$, though significantly weaker than the typical assumption of
independence of exogenous errors~\citep{pearl2000causality}. Intuitively, the
key assumption necessary for identification of the NDE is the absence of
a direct pathway between the unmeasured confounders and the outcome.

In CoVPN 3008, Assumptions~\ref{ass:no-V-to-Y} and~\ref{ass:equalE} are
plausible when quantifying the causal effect of hybrid versus vaccine immunity
$A$ through pathways involving a candidate immune correlate of protection $Z$
(the indirect effect) and those avoiding it (the direct effect) on a COVID-19
disease outcome $Y$, wherein the pathway between hybrid/vaccine immunity and the
candidate immune correlate may be confounded by factors $V$ that may include
pre-vaccination (baseline) levels of the immune marker or prior infection
history. In this setting, such confounding factors would be expected to only
impact the post-vaccination COVID-19 disease outcome through their influence on
hybrid/vaccine immunity and the post-vaccination immune response measured as
a candidate immune correlate of protection. Thus, use of
Assumptions~\ref{ass:no-V-to-Y} and~\ref{ass:equalE} to establish
identification of the NDE can facilitate interrogation of candidate immune
correlates of protection in CoVPN 3008, despite the likely presence of such
unmeasured confounders.


\subsection{Identifying Functional of the Natural Direct
Effect}\label{subsec:id-NDE}

\begin{theorem}\label{thm:id-NDE}
Consider the following complete data (causal) estimand,
\begin{equation*}
  \Psi^F_\text{NDE}(\Pt_{(U,X),0}) = \int_{\mathcal{W},\mathcal{Z}} \E[Y(1,z) - Y(0,z)
    \mid W=w] g_Z(z \mid w, A=0) p_W(w)~dz~dw \ .
\end{equation*}
Under Assumptions~\ref{ass:posA}, \ref{ass:posZ},
\ref{ass:no-V-to-Y} and~\ref{ass:equalE}, the corresponding observed data
estimand may be expressed as
\begin{align*}
\psi_\text{NDE} = \Psi_\text{NDE}(\Pt_0) &= \int_{\mathcal{W}} \int_{\mathcal{Z}}
    \left(\overline{Q}_Y(w,A=1,z) - \overline{Q}_Y(w,A=0,z)\right)
    g_Z(z \mid w, A=0)~dz~p_W(w)~dw\\
  &= \E_{\Pt_0}\left[\E_{\Pt_0}\left[\E_{\Pt_0}[Y \mid W,A=1,Z] -
    \E_{\Pt_0}[Y \mid W,A=0,Z] \mid W,A=0\right]\right] \ .
\end{align*}
\end{theorem}
A proof is provided in Section~\ref{sec:nde-ident} of the Appendix. We again emphasize that this identification
result is immediately applicable to the RR NDE of
Equation~\eqref{eqn:rr_nat_decomp}.

\begin{corollary}
Consider the following complete data (causal) estimand:
\begin{equation*}
    \Psi_\text{RR-NDE}^F(\Pt_{(U,X),0}) = \frac{\int_{\mathcal{W},\mathcal{Z}} \E[Y(1,z) 
    \mid W=w] g_Z(z \mid w, A=0) p_W(w)~dz~dw}{\int_{\mathcal{W},\mathcal{Z}} \E[Y(0,z)\mid W=w] g_Z(z \mid w, A=0) p_W(w)~dz~dw} \;.
\end{equation*}
Under the conditions outlined in Theorem~\ref{thm:id-NDE}, the corresponding observed data estimand may be expressed as
\begin{equation*}
    \psi_\text{RR-NDE} = \Psi_\text{RR-NDE}(\Pt_0) = \frac{\E_{\Pt_0}\left[\E_{\Pt_0}\left[\E_{\Pt_0}[Y \mid W,A=1,Z]\mid W,A=0\right]\right]}{\E_{\Pt_0}\left[\E_{\Pt_0}\left[
    \E_{\Pt_0}[Y \mid W,A=0,Z] \mid W,A=0\right]\right]} \;.
\end{equation*}
\end{corollary}

Note that the discussion in Section~\ref{causal_sem} focused on the NDE as the
parameter of interest, while Equation~\eqref{eqn:nde_stoch} introduced the
alternative characterization of the NDE as a contrast between the counterfactual
mean outcomes when the exposure is deterministically intervened upon (contrast
setting $A = 1$ and setting $A = 0$) and the mediators are drawn from the CDF of
$g_Z(z \mid W, 0)$. The complete data parameter introduced in
Theorem~\ref{thm:id-NDE} corresponds to an average of controlled direct effects
(CDEs), where the average is over possible values of the mediators $Z$; we
discuss identification of the CDE in Section~\ref{sec:cde-ident} of the Appendix. The principal difference between the NDE and
the average of CDEs can be attributed to the intervention drawing counterfactual
values of the mediators using $g_Z(z \mid W,0)$---for the NDE, this
post-intervention conditional density corresponds to the counterfactual density
of $Z(0) \mid W$, as opposed to $g_{Z}(z \mid W,0)$. While this may initially
appear unsatisfactory, we emphasize that, under the consistency assumption for
the NDE, $\P_{\Pt_{(U,X),0}}(Z(0) < z \mid W) = \P_{\Pt_0}(Z < z \mid W, A =
0)$. With this in mind, Theorem~\ref{thm:id-NDE} in fact constitutes an
identification result for the NDE, one which relaxes the required assumptions
by allowing for unmeasured confounding of the exposure--mediator relationship.
This relaxation of the standard assumptions on no unmeasured confounding only
requires that the pathway by which any unmeasured confounders $V$ affect $Y$ be
blocked by $Z$. An interesting---and, indeed, reassuring---property of this
result is that, when there are no unmeasured confounders (that is, $V \equiv
\emptyset$), the original identification results
of~\citet{robins1992identifiability} and~\citet{pearl2001direct} are recovered.

\section{Inference}\label{sec:inference}

Having presented our result for identification of the (RR) NDE under unmeasured
exposure--mediator confounding, we now make recommendations for performing
inference about $\Psi_\text{NDE}$. We emphasize that inference about
$\psi_\text{RR-NDE}$ follows immediately from a straightforward application of
the functional delta method \citep{vdl2011targeted}. Estimation approaches
appearing in the non-exhaustive summary that follows are compatible with our
relaxed identification strategy.

Myriad inferential approaches have been developed for the NDE. The parametric
G-computation method, which yields a substitution estimator relying on maximum
likelihood estimators of the marginal covariate distribution, the conditional
mediator density, and the conditional outcome regression, can be unbiased when
all nuisance parameters are estimated
consistently~\citep{petersen2006estimation, imai2010general}. This estimator is
biased, however, if any of the nuisance parameter estimators are not consistent
for their respective target functionals and, further, is incompatible with the
use of nonparametric regression strategies for nuisance parameter estimation. The
estimating equation approach of~\citet{vdl2008direct}, building upon the
general framework of~\citet{vdl2003unified}, is robust to misspecification of
the conditional outcome regression or exposure assignment mechanism but
requires that the user-supplied conditional mediator density corresponds to the
(true) conditional mediator density under $\Pt_0$. A targeted maximum
likelihood (TML) estimation~\citep{vdl2006targeted, vdl2011targeted} approach
is presented in~\citet{vdl2008direct}, but this suffers from the same
limitation. \citet{zheng2012} proposed a TML estimator exhibiting a form of
multiple robustness based on the efficient influence function (EIF) of the NDE
derived by~\citet{tchetgen2012semiparametric}, who proposed an efficient
estimator based on one-step bias-correction~\citep{pfanzagl1985contributions,
bickel1993efficient}. The one-step and TML estimators are both asymptotically
unbiased if any of the following conditions hold: (i) the conditional expected
outcome given confounders, exposure, and mediators, and the
confounder-conditional expected difference of the conditional expected outcomes
given confounders, mediators, and different contrasts of the exposure are
consistently estimated; (ii) the exposure mechanism and the conditional outcome
given confounders, exposure, and mediators are consistently estimated; (iii)
the exposure mechanism and the mediator density conditioning on the exposure
and confounders are consistently estimated; or (iv) the exposure mechanism and
the conditional distribution of the exposure given confounders and of the
mediators given exposures and confounders are consistently estimated.
Furthermore, if conditions (i) and (ii), and either of (iii) or (iv) hold, then
these estimators will be asymptotically efficient, achieving the
semi-parametric efficiency bound that applies to all regular asymptotically
linear estimators in the nonparametric model $\M$. Despite their asymptotic
equivalence, only the TML estimator is guaranteed to produce estimates that
respect global model constraints, by virtue of it being a plug-in estimator.

For completeness, we briefly present the one-step and TML
estimators of the NDE, as well as their asymptotic properties
\citep[see][for more details]{tchetgen2012semiparametric, zheng2012}. Note that the
EIF characterizes the scaled difference of the
estimator of the target parameter evaluated at $\hat{\Pt}_n$---a plug-in
estimator of $\Pt_{0}$ made up of elements of $\Pt_{n}$, and possibly augmented
by nuisance parameter estimators---and centered at $\Psi_\text{NDE}(\Pt_0)$
through a von Mises expansion~\citep{mises1947asymptotic, bickel1993efficient,
vdvaart2000asymptotic, vdl2003unified, hines2022, kennedy2022semiparametric}:
\begin{equation}\label{eq:von-mises}
  \begin{split}
    \sqrt{n}\left(\Psi_\text{NDE}(\hat{\Pt}_n)-\Psi_\text{NDE}(\Pt_0)\right)
    & = \frac{1}{\sqrt{n}}\sum_{i=1}^n D^\star(O_i, \Pt_0) -
      \frac{1}{\sqrt{n}}\sum_{i=1}^n D^\star(O_i, \hat{\Pt}_n) \\
    & + \sqrt{n}\left(\E_{\Pt_n} - \E_{\Pt_0}\right)
      \left(D^\star(O,\hat{\Pt}_{n}) - D^\star(O, \Pt_0)\right) -
      \sqrt{n} R(\Pt_0, \hat{\Pt}_n) \ ,
\end{split}
\end{equation}
where $\sqrt{n} R(\Pt_0, \hat{\Pt}_n)$ is an exact second-order remainder term. By the
central limit theorem and moment properties of the EIF, the first term converges
to a Gaussian distribution with mean zero and variance equal to
$\E_{\Pt_0}[D^\star(O, \Pt_0)^2]$. The second term of
Equation~\eqref{eq:von-mises} is a plug-in bias term that, in general, does not
vanish asymptotically. The one-step estimator is obtained by adding this term to
the plug-in estimator. For the TML estimator,~\citet{zheng2012} rely on the TML
estimation framework~\citep{vdl2006targeted, vdl2011targeted} to modify initial
nuisance parameter estimates such that this term is negligible when using these
updated nuisance estimators within the plug-in estimator. That is, in TML
estimation, the initial nuisance parameter estimators are tilted to also be
a solution to the EIF estimating equation;
for conciseness, we omit the details here but include them in
Section~\ref{sec:tmle} of the Appendix. The third
and fourth terms converge to zero in probability under standard empirical
process conditions and if any of the previously listed consistency requirements
on nuisance parameter estimators are satisfied~\citep{zheng2012}. Note that the
empirical process conditions can be loosened by applying sample-splitting
(cross-fitting) in nuisance parameter
estimation~\citep{schick1986asymptotically, klaassen1987consistent,
zheng2011cross, chernozhukov2018double}.

Denoting the TML estimator of~\citet{zheng2012} by
$\Psi_\text{NDE}(\Pt_n^\star)$, where $\Pt_n^\star$ is a plug-in estimator of
$\Pt_0$ made up of $\Pt_n$ and TML-adjusted nuisance parameter estimators, and
noting that the EIF of the target parameter $\Psi_\text{NDE}(\Pt)$ defined in
Theorem~\ref{thm:id-NDE} is
\begin{equation}\label{eq:eif}
  \begin{split}
    D^\star(O, \Pt)
    & = \left(\frac{\I(A=1)}{g_A(W)}
      \frac{g_Z(Z \mid W,0)}{g_Z(Z \mid W, 1)}
      - \frac{\I(A=0)}{1 - g_A(W)}\right)
      \left(Y - \overline{Q}_Y(W, A, Z)\right) \\
    & \,\,\,\, + \frac{\I(A=0)}{1 - g_A(W)}
    \left(\overline{Q}_Y(W, 1, Z) - \overline{Q}_Y(W, 0, Z) -
      \E_{\Pt}\left[\overline{Q}_Y(W, 1, Z) -
      \overline{Q}_Y(W, 0, Z)\big|W, A=0\right]
    \right)\\
    & \,\,\,\, + \E_{\Pt}\left[\overline{Q}_Y(W, 1, Z) -
      \overline{Q}_Y(W, 0, Z)\big|W, A=0\right] - \Psi_\text{NDE}(\Pt) \ ,
  \end{split}
\end{equation}
we find, under the aforementioned sufficient conditions, that
\begin{equation*}
  \sqrt{n}\left(\Psi_\text{NDE}(\Pt_n^\star) - \Psi_\text{NDE}(\Pt_0)\right) =
  \frac{1}{\sqrt{n}}\sum_{i=1}^{n} D^{\star}(O_i, \Pt_0) + o_{\Pt}(1) \ .
\end{equation*}
This implies that
\begin{equation*}
  \Psi_\text{NDE}(\Pt_n^\star) \overset{D}{\rightarrow}
  N(\Psi_\text{NDE}(\Pt_0), \E_{\Pt}[D^\star(O, \Pt_0)^2]/n) \ .
\end{equation*}

That is, assuming that the identification result outlined in
Section~\ref{sec:identification} holds, $\Psi_\text{NDE}(\Pt_n^\star)$ is a
consistent---and potentially asymptotically efficient---estimator of the NDE.
When its asymptotic linearity conditions are satisfied, Wald-type confidence
intervals (CIs) can be constructed from the EIF-based variance estimator.
Similarly, formal hypothesis testing may be performed. These results hold for
the one-step estimator of $\Psi_\text{NDE}(\Pt_0)$ as well.

In our motivating setting, evaluating immune correlates of protection in the
CoVPN 3008 study, candidate immune correlates are measured via a case-cohort
sampling design, a form of \textit{outcome-dependent}, two-phase
sampling~\citep{breslow1997maximum, breslow2000semi, breslow2003large} commonly
used to measure immune responses in vaccine efficacy
trials~\citep{follmann2006augmented}. Two-phase sampling necessitates
adjustments to the NDE estimator and EIF alike~\citep{robins1994estimation,
rose2011targeted2sd, hejazi2020efficient}. We provide details on this in
Section~\ref{sec:two-phase} of the Appendix.

Open-source implementations of the cross-fitted NDE and RR NDE estimators are
available in the \texttt{medoutcon} software
package~\citep{hejazi2022medoutcon-rpkg, hejazi2022medoutcon-joss} for the
\texttt{R} language and environment for statistical computing~\citep{R}. The
Super Learner method~\citep{vdl2007super} implemented in the \texttt{sl3}
\texttt{R} package~\citep{coyle2020sl3} is used for nuisance parameter
estimation. This permits the use of flexible, data-adaptive regression and
machine learning procedures throughout the estimation process, reducing
opportunities for model misspecification. The \texttt{medoutcon} \texttt{R}
package is used in the simulation study presented in Section~\ref{sec:sim}, the
simulated comparative vaccine study in Section~\ref{sec:applic-prospect}, and
the application to CoVPN 3008 in Section~\ref{sec:applic-covpn3008}.

\section{Simulation Study: Impact of Unmeasured Exposure--Mediator
Confounding}\label{sec:sim}

Practical implications of the identification result presented in
Theorem~\ref{thm:id-NDE} for inference about the NDE are explored
with simulated observational study data generated according to following
data-generating process:
\begin{align*}
  W_1 & \sim \text{Unif}(-1, 1);
  W_2 \sim N(0, 1);
  V \sim N(0, 1) \\
  A \mid W, V & \sim \text{Bern}
    \left(\left(1 + \exp\{-W_1-W_2-V\}\right)^{-1}\right) \\
  Z \mid A, W, V & \sim \text{Bern}
    \left(\left(1 + \exp\{-W_1-W_2- \gamma V - 3A\}\right)^{-1}\right) \\
  Y \mid A, W, Z & \sim N(3A + W_1 + W_2 + Z, 1) \;.
\end{align*}

Here, $W = (W_1, W_2)$ are baseline confounding variables, $V$ is an
unobserved exposure--mediator confounder, $A$ is an indicator of assignment to
the exposure condition, $Z$ is a mediator, and $Y$ is the outcome. We highlight
that this data-generating process satisfies Assumptions~\ref{ass:no-V-to-Y}
and~\ref{ass:equalE}: the effect of $V$ on $Y$ operates solely through $Z$.

We consider four effect size values for the effect of the unmeasured confounder
$V$ on the mediator $Z$, $\gamma$: $0, 1, 2$ and $3$. For each value of $\gamma$,
the influence of unmeasured confounding on the mediator is assessed via the relative
risk of $Z \mid V$ obtained by fixing $V$ at its fourth and first quartile. These
relative risks are denoted by $\text{RR}_Z$, and are computed via Monte Carlo simulation.
The $\text{RR}_Z$ for $\gamma$ equal to $0, 1, 2$, and $3$ are $1.335$, $2.166$,
$3.379$, and $5.215$, respectively. These values
were computed through the numerical integration of these parameters' closed-form
expressions~\citep{cubature}. This implies that the probability that
$Z=1$, conditional on $V$ and $A$, increases as $\gamma$ increases.

\begin{figure}[ht!]
  \centering
  \includegraphics[width=0.95\textwidth]{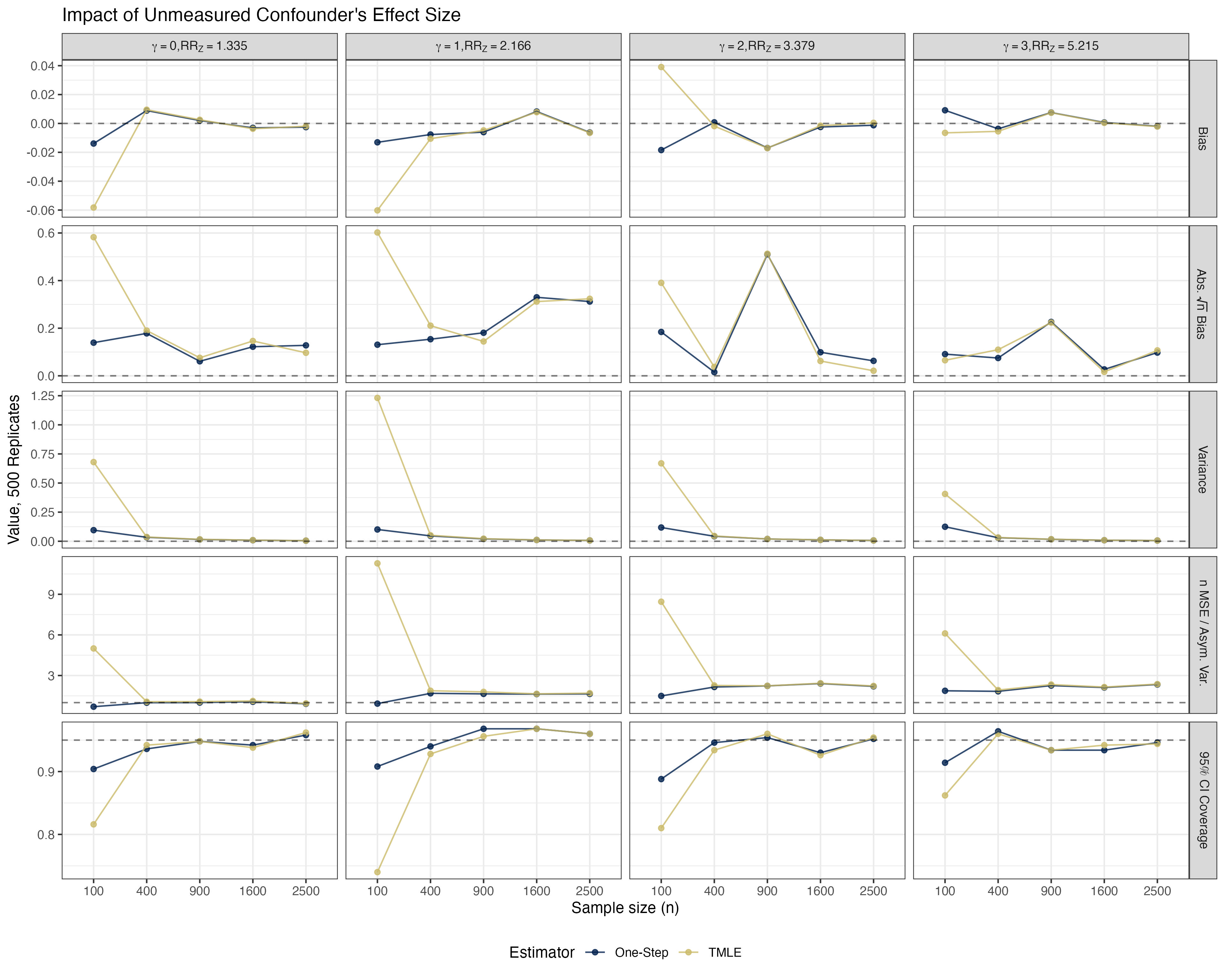}
  \caption{\emph{Simulation study: Impact of unmeasured exposure--mediator confounding.} The dashed horizontal lines represent the desired asymptotic values
    of each metric. These values are zero for bias and variance, one for the
    normalized MSE, and 0.95 for coverage of 95\% Wald-type CIs. The dashed line
    for the normalized MSE corresponds to the asymptotic efficiency bound when
    there is no unmeasured confounding $(\gamma=0)$.}
  \label{fig:sim-results}
\end{figure}

Varying across samples of size $n=$ $100$, $400$, $900$, $1,\!600$, and
$2,\!500$ (which covers the sample size early-phase vaccine efficacy studies),
we assessed the finite-sample impact of $\gamma$'s magnitude on inference for
the NDE using the one-step and TML estimators described in
Section~\ref{sec:inference}. All nuisance parameter estimators were estimated
using GLMs, except for $\E_{P}[\overline{Q}_Y(W,1,Z) -
\overline{Q}_Y(W,0,Z)\big|W,A=0]$, which was estimated using the highly adaptive
lasso (HAL) of~\citet{vdl2017generally} via the \texttt{hal9001} \texttt{R}
package~\citep{coyle2020hal9001, hejazi2020hal9001-joss}. Five hundred random
samples were generated for each sample size and value of $\gamma$ from the
data-generating process above. The empirical bias, absolute scaled empirical
bias, variance, normalized mean squared error (MSE), and coverage of 95\%
Wald-type CIs were then computed for each scenario
(Figure~\ref{fig:sim-results}). The \texttt{simChef} \texttt{R} package provided
the architecture for this simulation study~\citep{simChef}.

We find that the magnitude of $\gamma$ has no appreciable impact on empirical
bias metrics, variance, or test coverage. In all sample sizes larger than
$n = 100$, the estimators exhibit minimal bias and the accompanying inferential
procedures produce CIs that realize the nominal coverage rate of 95\%. While the
one-step and TML estimators behave similarly, the one-step estimator's
performance appears superior in smaller sample sizes---it is approximately
unbiased even when $n=100$, its empirical variance is no larger than $0.125$,
and its empirical coverage is above $\approx90$\% in all but one
scenario---which is unexpected given that it does not obey the substitution
principle. The asymptotic efficiency of the estimators, however, decreases as
$\gamma$ increases regardless of sample size.

That the estimators fail to achieve the same efficiency bound set by $\gamma=0$
when $\gamma=1,2$ or $3$ is unsurprising. The identification result of
Theorem~\ref{thm:id-NDE} guarantees that a familiar statistical estimand would
correspond to the NDE under a novel set of relaxed assumptions, which allow
a form of unmeasured confounding by $V$, \textit{not} that the corresponding
estimator would be asymptotically efficient under such unmeasured confounding.
Indeed, failing to account for relevant sources of variation---an unmeasured
confounder $V$---can plausibly lead to estimators with inflated variance, which
is to say \textit{there ain't no such thing as a free lunch}. While unmeasured
confounding usually compromises inference by leading to non-identification of
target parameters, here, we have succeeded, through carefully formulated
relaxed identification conditions, in achieving identification and limiting
its subsequent impact on reduced efficiency.

\section{Simulation Study: Application in a Vaccine Prospective
Cohort Study}\label{sec:applic-prospect}

Our inferential approach was developed in anticipation of its application to
data from the CoVPN 3008 study, with most developments occurring concurrently
with the study's execution. We engineered a data-generating process expected to
capture key characteristics of the data collected in CoVPN 3008 to assess the
operating characteristics of the cross-fitted RR NDE estimators adjusted for two-phase
sampling. Trial participants, defined as $O = (W, A, R, RZ, Y)$, were generated
\begin{align*}
  V & \sim \text{round}(\text{Unif}(0.1, 0.8), 1) \\
  W_1 & \sim \text{Bern}(0.3) \, , \,
  W_2 \sim \text{Bern}(0.3) \, , \,
  W_3 \sim \mathbb{I}(\text{Poisson}(30) + 20 > 50) \\
  A \mid W, V & \sim \text{Bern}\left(\left(1 + \exp\{0.5 - 0.5 W_1 + 0.5 W_2 -
    0.2 W_3 - 0.2 V\}\right)^{-1}\right) \\
  Z \mid A, W, V & \sim N \left(5.661 + 0.3 W_1 - 0.7 W_2 - 0.5 W_3 + 0.3 V -
    A, 1\right) \\
  Y \mid A, W, Z & \sim \text{Bern}\left(\left(1 + \exp\{-2.5 W_1 - 0.3 W_2 -
    0.4 W_3 + 0.5 A - 0.2 Z\}\right)^{-1}\right) \\
  R \mid Y & \sim
             \begin{cases}
               1, & Y = 1 \\
               \text{Bern}(\eta), & Y = 0
             \end{cases} \;.
\end{align*}
Here, $W, A, Z$, and $Y$ are defined as before, and $R$ is an indicator of
inclusion in the phase-two (case-cohort) sample. $V$ represents the unknown
deciles of baseline candidate immune correlate of protection measurements,
$\text{round}(\cdot)$ corresponds to rounding to the nearest tenth, and $\eta$
is the proportion of observations randomly selected at baseline for
post-vaccination measurement of the candidate immune correlate of protection in
the phase-two sample created by the case-cohort sampling design. We considered
four possible scenarios for $\eta$, letting $\eta$ equal $0.25$, $0.5$, $0.75$,
or $1$. The value of $1$ for $\eta$ is equivalent to the setting without
two-phase sampling, meaning all participants have their candidate immune
response mediators measured.

We compared the performance of four RR NDE candidate estimators: one-step
(``One-Step (Obs.~Weights)'') and TML (``TMLE (Obs.~Weights)'') estimators that
treat the two-phase sampling probabilities as known by design as well as
one-step (``One-Step'') and TML (``TMLE'') estimators that additionally accounts
for estimation of the outcome-dependent sampling mechanism. The One-Step
(Obs.~Weights) and TMLE (Obs.~Weights) estimators are equivalent to the one-step
and TML estimators of the RR NDE when two-phase sampling is not performed, save
that they are re-weighted by the inverse probability of selection into the
phase-two sample. Simply put, two-phase sampling weights are used as observation
weights by these estimators of the RR NDE. One-Step (Obs.~Weights) and TMLE
(Obs.~Weights) estimators are therefore expected to be less efficient than the
One-Step and TMLE estimators since they ignore information about observations
for which the candidate immune correlate were not collected. These estimators'
nuisance parameters were estimated using a Super Learner with library composed
of GLMs, Bayesian GLMs, penalized GLMs
\citep{hoerlRidgeRegressionBiased1970,
tibshiraniRegressionShrinkageSelection1996,
zouRegularizationVariableSelection2005}, and random forests
\citep{breiman2001random}.

For each level of $\eta$, we drew $500$ dataset replicates for each of $n =
1,\!000, 2,\!000, 4,\!000$, and $8,\!000$. Across each scenario and sample size,
we measured the empirical bias, scaled absolute bias, variance, and 95\% CI
coverage of the estimators. The 95\% CIs of $\psi_\text{RR-NDE}$ were obtained
by exponentiating the lower and upper bounds of the 95\% CIs of
$\log\psi_\text{RR-NDE}$, ensuring that the resulting interval estimates
respected the bounds of the parameter space. The results are presented in
Figure~\ref{fig:nde-inference-simulations}.

\begin{figure}[ht!]
  \centering
  \includegraphics[width=\textwidth]{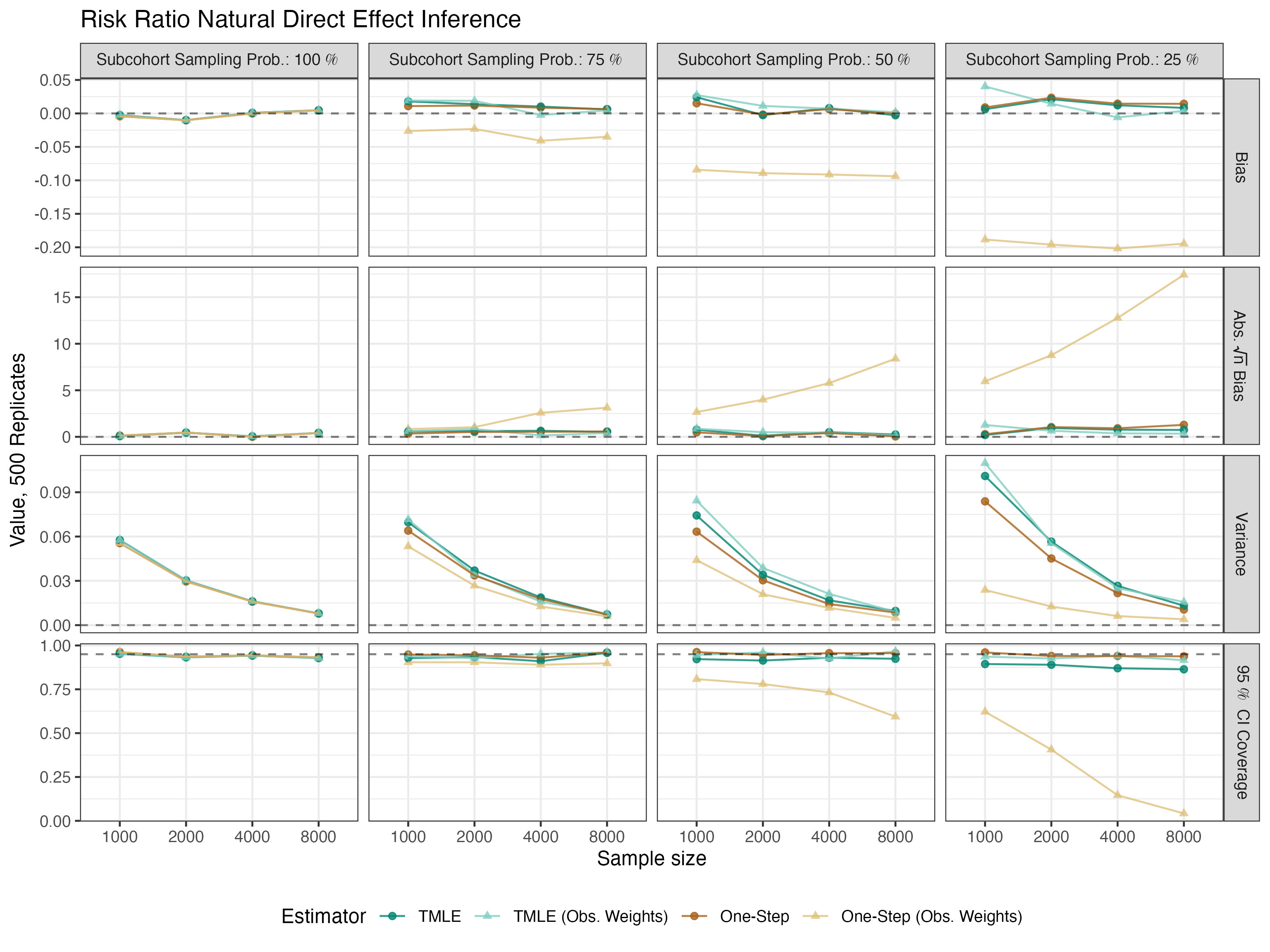}
  \caption{\emph{Simulation study: Application in a vaccine prospective
    cohort study.} The dashed lines indicate the desired asymptotic values of
    each metric, as in Figure~\ref{fig:sim-results}. The \textit{One-Step} and
    \textit{TMLE} estimators correspond to the estimators that explicitly
    estimate the two-phase sampling mechanism. The \textit{One-Step (Obs.
    Weights)} and \textit{TMLE (Obs. Weights)} estimators correspond to the
    estimators that incorporate the two-phase sampling weights as observations
    weights.}
  \label{fig:nde-inference-simulations}
\end{figure}

As expected, all estimators behaved similarly when $\eta = 1$ --- that is, when
the mediators of all study participants are measured. In all other scenarios, the
One-Step (Obs. Weight) estimator was empirically biased, and this empirical bias
increases as $\eta$ decreases. This poor performance can be attributed to the
numerator and denominator making up the true RR NDE being near the bounds of
the parameter space. One-step estimators may be sensitive to this as they are not
plug-in estimators; they are more likely to exhibit numerical instability. The
remaining estimators performed similarly with respect to empirical bias, scaled
absolute empirical bias, variance, and coverage of their 95\% confidence
intervals. With that said, the One-Step estimator had slightly lower empirical
variance than the other approximately unbiased estimators, and the TML
estimator exhibited slight under-coverage when $\eta = 0.25$.

\section{Application to the CoVPN 3008
Study}\label{sec:applic-covpn3008}


Our developments were motivated by their prospective application to the CoVPN
3008 study. The CoVPN 3008 (Ubuntu) study (NCT05168813) is a recent phase 3,
multi-center, multi-stage non-randomized clinical study of the mRNA-1273 vaccine
that enrolled adult study participants, aged 18 years or older, who were living
with HIV and/or other comorbidities found to be associated with severe COVID-19
disease in seven Eastern and Southern African countries. Prior infection by
SARS-CoV-2 was assessed using a combination of two SARS-CoV-2 serology tests and
a SARS-CoV-2 nasal swab in order to assign participants in a non-randomized
fashion into hybrid and vaccine immunity groups, with the former receiving one
dose of the study vaccine and the latter receiving two doses. The primary
endpoint was the first occurrence of COVID-19 or severe COVID-19 disease
(U.S.~CDC definition) after final vaccine dose through 165 days after the final
vaccination. Safety and efficacy results were reported
by~\citet{garrett2025hybrid} and immune correlates analyses
by~\citet{mkhize2025neutralizing}.

A secondary aim of CoVPN 3008 was to evaluate how a previously studied immune
correlate of protection, the neutralizing antibody titer (ID50,
$\log_{10}$-scaled) against the BA.4/BA.5 variants of SARS-CoV-2, may vary
between hybrid and vaccine immunity groups with respect to the COVID-19
disease endpoint of interest by six months post-vaccination. As the hybrid and
vaccine immunity groups are defined by receipt of one versus two doses of the
study vaccine, respectively, and evidence of prior infection, we let $A=1$
denote receipt of one vaccine dose and evidence of prior infection (hybrid
immunity) and $A=0$ denote receipt of two vaccine doses and no evidence of
prior infection (vaccine immunity). The candidate immune correlate of
protection is the measurement of the neutralizing antibody titer (ID50) against
the SARS-CoV-2 BA.4/BA.5 variants at Day 29 post-vaccination, which we denote
by $Z$. The outcome, denoted $Y$, is the indicator of COVID-19 or severe
COVID-19 disease (U.S.~CDC definition) starting 7 days after the Day 29 visit
post final vaccination by 165 days following receipt of the final vaccine. We
take the unmeasured exposure--mediator confounder, denoted $V$, to be IgG
concentration against the N protein, which is informative of infection history
and infection recency; in CoVPN 3008, it is a phase-two variable, measured on
only a small subset of study participants. For this choice of $V$, our immune
correlates analysis is robust to unmeasured confounding by IgG concentration
against the N protein, which is not generally expected to be measured in the
plurality of study participants in a preventive vaccine study with a rare
primary outcome. As such, we do not use data on $V$ collected in study
participants in the phase-two sample. For this analysis, we adopt an adjustment
set defined by~\citet{garrett2025hybrid} for their reported analysis, consisting
of enrollment region, enrollment period, baseline HIV status, baseline TB
status, and a baseline infection risk score; we denote these $W$.

We report our findings in terms of RR NDE and consider an informal sensitivity
analysis designed to examine the degree of confounding necessary to undermine
any previously posited notion of mediation of the beneficial effect of hybrid
versus vaccine immunity on COVID-19 disease through Day 29 neutralizing antibody
titer (ID50), previously validated in vaccine--placebo comparisons in the
context of vaccine efficacy trials~\citep{gilbert2022covid}. Specifically, we
test the following hypothesis for a fixed value of $\psi_{\text{RR}} < 1$
representing a \textit{protective} total effect:
$H_0: \psi_{\text{RR-NDE}} \leq \psi_{\text{RR}}$ versus
$H_1: \psi_{\text{RR-NDE}} > \psi_{\text{RR}}$.
Under $H_1$, it follows that $\psi_{\text{RR-NIE}} < \psi_{\text{RR}} <
\psi_{\text{RR-NDE}}$, suggesting evidence that the postulated protective
effect operates, at least in part, through neutralizing antibody titer. In our
sensitivity analysis, we consider estimation of $\psi_{\text{RR-NDE}}$ based on
our novel identification result of the NDE, which holds regardless of possible
exposure--mediator confounding by $V$; however, $\psi_{\text{RR}}$ is based on
an estimate derived from results reported by~\citet{garrett2025hybrid}, arrived
at via the standard set of no unmeasured confounding assumptions previously
described in Section~\ref{subsec:ident:ss}.
As such, we consider a range of plausible values for $\psi_{\text{RR}}$,
anchored around the estimate derived from~\citet{garrett2025hybrid}, with the
goal of considering how much $\psi_{\text{RR}}$ would need to differ from the
reported estimate---due to unmeasured confounding---in order for available
evidence to be insufficient to reject $H_0$. We operationalized our sensitivity
analysis using the cross-fitted case-cohort TML estimator from
Section~\ref{sec:inference}, denoted $\hat\psi_\text{RR-NDE}$, using a super
learner ensemble~\citep{vdl2007super} with a library of candidate algorithms
composed of regularized regression (lasso and ridge), multivariate adaptive
regression splines, and two variants of random forests with different numbers of
trees. Two-phase sampling weights based on the case-cohort design were
incorporated, and the analysis included $n=310$ phase-two participants, among
whom $n=122$ experienced the COVID-19 endpoint by month six post-vaccination.
Additional details on the case-cohort design have previously been reported by
\citet{mkhize2025neutralizing}. 

Our analysis yielded an estimate $\hat\psi_\text{RR-NDE} = 0.570$,
whereas~\citet{garrett2025hybrid} reported a $\psi_\text{RR}$ point estimate of
$0.613$ when comparing hybrid and vaccine immunity groups; the ordering of the
two point estimates is inconsistent with neutralizing antibody titer mediating a
portion of hybrid immunity's protective effect. Further, the one-sided test of
$H_0: \psi_{\text{RR-NDE}} \leq 0.613$ versus $H_1: \psi_{\text{RR-NDE}} >
0.613$ yields a test statistic
$t = (\log(\hat\psi_{\text{RR-NDE}}) - \log(0.613)) /
\hat{\sigma}(\log(\hat\psi_{\text{RR-NDE}})) \approx -0.097$, where
$\hat{\sigma}(\log(\hat\psi_{\text{RR-NDE}}))$ is the estimated standard
deviation of the efficient influence function of
$\log(\hat\psi_{\text{RR-NDE}})$. Assuming conditions for the asymptotic
linearity of $\hat\psi_\text{RR-NDE}$ are satisfied such that $t$ is normally
distributed under the null, we find that this test statistic corresponds to a
p-value of $0.539 > 0.1$, thereby failing to reject $H_0$. However, the $95\%$
Wald-type confidence interval of $\psi_{\text{RR-NDE}}$ is $(0.186; 1.743)$,
indicating substantial uncertainty about the target parameter's underlying true
value. It is impossible to rule out the possibility of a protective effect.
\begin{figure}[!ht]
  \centering
  \includegraphics[width=\textwidth]{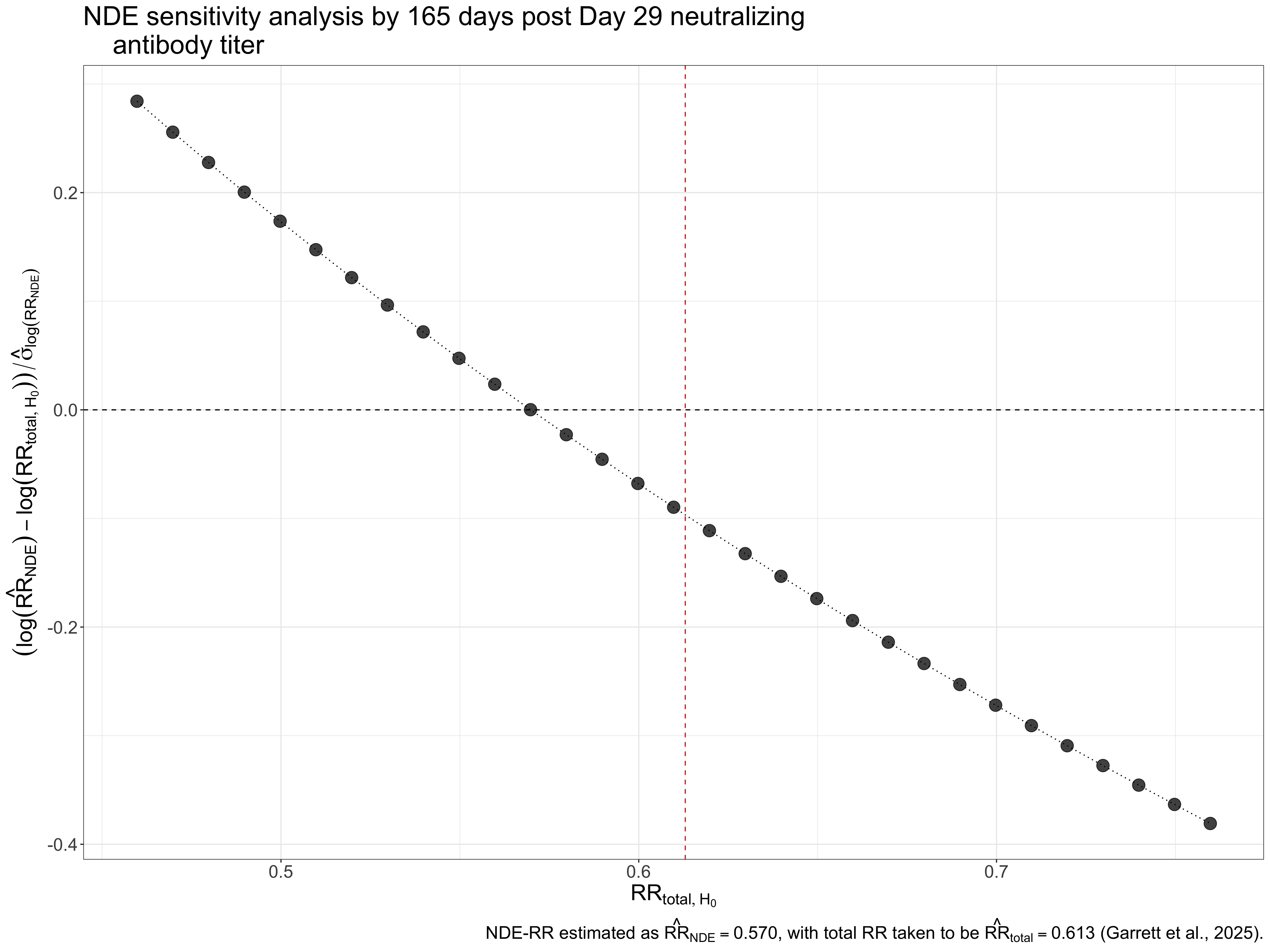}
  \caption{\emph{Application in the CoVPN 3008 study.} Test statistic $t$
  (y-axis) measuring the degree of deviation from $H_0$ versus plausible more
  protective (to left) and less protective (to right) total $\psi_\text{RR}$
  values (x-axis). The red vertical dashed line corresponds to the value of
  $\psi_\text{RR}$ considered as $H_0$ in our primary analysis. The horizontal dashed
  line corresponds to the null hypothesis of no protective effect.}
  \label{fig:nde-sensitivity}
\end{figure}%
Figure~\ref{fig:nde-sensitivity} presents a sensitivity
analysis showing the test statistic $t = (\log(\hat\psi_{\text{RR-NDE}}) -
\log(\psi_{\text{RR}})) / \hat{\sigma}(\log(\psi_{\text{RR-NDE}}))$ for plausible
values of $\psi_{\text{RR}}$. The vertical dashed line indicates
$\psi_{\text{RR}} = 0.613$ from~\citet{garrett2025hybrid}; the horizontal dashed
line at $t = 0$ corresponds to $\text{RR}_{\text{NDE}} = \psi_{\text{RR}}$
(no indirect effect). The test statistic remains below conventional significance
thresholds across the examined range, though some $\psi_\text{RR}$ point
estimates suggest evidence of mediation.

In the context of our methodological contributions, of which a key innovation is
permitting identification of the NDE under unmeasured confounding of the
exposure--mediator relationship, we obtain valid inference about the NDE even
when unmeasured confounding by IgG concentration against the N protein may cast
doubt upon stronger and more typical no unmeasured confounding assumptions. Such
assumptions have  necessarily been invoked in prior immune correlates
analyses~\citep{mkhize2025neutralizing} of CoVPN 3008.

\section{Discussion}\label{sec:discussion}

We have outlined relaxed, nonrestrictive conditions for identification of the
NDE from observational data in the presence of unmeasured exposure--mediator
confounding. In addition to the standard identification
assumptions---exposure--outcome (\ref{ass:original_A1}) and mediator--outcome
(\ref{ass:original_A3}) conditional randomization, and positivity
(\ref{ass:posA} and \ref{ass:posZ})---we show that the NDE is identified so long
as no unmeasured endogenous pathway exists between unmeasured exposure--mediator
confounders and the outcome (\ref{ass:no-V-to-Y}), and that the
conditional expectation of the outcome does not depend on these confounders
(\ref{ass:equalE}). Put
another way, identification of the NDE is possible when the unmeasured
confounders' effect on the outcome operates exclusively through the mediators.

Opportunities for using this methodology are numerous. One example, which
motivated our work and is considered in the application of Section
\ref{sec:applic-covpn3008}, lies in the study of how immune-mediated protection,
induced by both vaccination and prior infection (typically unmeasured), may
operate differentially through immune correlates of protection. We highlight
that our approach can be informative of mechanism under an identification
strategy more conservative than the standard. Other typically unmeasured
confounders are numerous---from genetic factors predisposing individuals to
heightened or suppressed immune responses, unmeasurable-at-baseline
immunocompromised status, or molecular features of immune system health, to name
but a few. Applications beyond the study of immune correlates of protection
include studies of the putative mechanisms by which environmental exposures
cause disease and those of the mechanisms of medical interventions on disease
despite missing or unmeasured information in electronic health records, among
myriad others.

Defining and evaluating causal effects along pathways informative of mechanism
is a longstanding problem in statistical science, and our novel identification
strategy effectively expands the array of settings in which we may learn about
how exposures confer protection or cause harm. Future work will focus on
exploring further public health applications.



\subsection*{Code}

Code for reproducing the simulation studies is publicly available on GitHub
at \url{https://github.com/PhilBoileau/pub_NDE-unmeasured-confounding}.

\subsection*{Acknowledgments}

PB gratefully acknowledges the support of the Fonds de recherche du Qu\'{e}bec
- Nature et technologies and the Natural Sciences and Engineering Research
Council of Canada. NSH is grateful for the prior support of the National Science
Foundation (award no.~DMS 2102840). IM is grateful for the support of the
Harvard Data Science Initiative and the National Institute of Allergy and
Infectious Diseases (award no.~R01 AI074345). MJvdL was partially supported by a
grant from the National Institute of Allergy and Infectious Diseases (award
no.~R01 AI074345). PBG is grateful for support from the National Institute of
Allergy and Infectious Diseases (award no.~R37 AI054165). The content is solely
the responsibility of the authors and does not represent the official views of
the National Institutes of Health.


\appendix

\section{Equivalence of the Conditional Expectation of Y with and without V}
\label{sec:equivE}

In this section, we elaborate on the conditions under which
Assumption~\ref{ass:equalE}, on the equivalence of the conditional
expectation of the outcome $Y$ with and without unmeasured confounders $V$,
holds. Note that Assumption~\ref{ass:equalE} implicitly places restrictions
on the joint distribution of the exogenous factors $U$, a significantly weaker
condition than the typical assumption of independence of exogenous
errors~\citep{pearl2000causality}. In the following, we propose a set of
additional assumptions necessary for Assumption~\ref{ass:equalE} to hold in
the system depicted by the DAG in Figure~\ref{fig:DAG}.

\vspace{0.5em}
\begin{assumptionb}[\textit{Complete-data treatment-outcome conditional
randomization}]\label{ass:restrict_UA}
$U_A \indep U_Y \mid (W,V)$.
\end{assumptionb}
\vspace{0.5em}
\begin{assumptionb}[\textit{Complete-data mediator-outcome conditional
randomization}]\label{ass:restrict_UZ}
$U_Z \indep U_Y \mid (W,V)$.
\end{assumptionb}
\vspace{0.5em}
\begin{assumptionb}[\textit{Complete-data outcome-unmeasured confounders
conditional independence}]\label{ass:restrict_UY}
$U_Y \indep V \mid W$.
\end{assumptionb}
\vspace{0.5em}
\begin{lemma}\label{lemma:equivE}
  Suppose Assumptions~\ref{ass:restrict_UA},~\ref{ass:restrict_UZ},
  and~\ref{ass:restrict_UY} hold. Under Assumption~\ref{ass:no-V-to-Y},
  \ref{ass:posA} and \ref{ass:posZ} we then have the following
  equivalence: $\E[Y \mid W, V, A = a, Z] = \E[Y \mid W, A = a, Z]$.
\end{lemma}

\begin{proof}
  Since $\E[Y(a) \mid W,V] = \E[Y \mid W, V, A=a, Z=z]$ (by consistency,
  definition of counterfactuals) and $A \in \{0,1\}$, it suffices to show that
  $\E[Y \mid W, V, A=1, Z=z]$ does not depend on $V$. To see this, consider the
  following:
\begin{align*}
  \E[Y \mid W=w, V, A=1, Z=z] &= \E[f_Y(w, V, 1, z, U_Y) \mid W=w, V, A=1, Z=z]
  ~\text{(by definition)}\\
  &= \E[f_Y(w, 1, z, U_Y) \mid W=w, V, A=1, Z=z]
  ~\text{(by Assumption~\ref{ass:no-V-to-Y})}\\
  &= \E[f_Y(w, 1, z, U_Y) \mid W=w, V, f_A(w, V, U_A) =1, f_Z(w,V,1, U_Z)=z]
  \\ &= \E[f_Y(w, 1, z, U_Y) \mid W=w, V]
  ~\text{(by Assumptions~\ref{ass:restrict_UA} and~\ref{ass:restrict_UZ})} \\
  &= \E[f_Y(w, 1, z, U_Y) \mid W=w]
  ~\text{(by Assumptions~\ref{ass:restrict_UY}, \ref{ass:posA},
   \ref{ass:posZ})}.
\end{align*}
\end{proof}
We emphasize that none of the assumptions necessary to verify
Assumption~\ref{ass:equalE} place any restrictions on the distribution of
the observed data unit $O=(W, A, Z, Y)$. Assumption~\ref{ass:restrict_UA} is a
treatment randomization assumption standard in the causal inference literature.
If $U_Y$ is independent of $U_A$ given $(W,V)$, then $Y(a)$ is independent of
$A$ given $(W,V)$, analogous to Assumption~\ref{ass:original_A1}.
Assumption~\ref{ass:restrict_UZ} implies independence of $U_Z$ and $U_Y$ given
$(W, V)$. If $U_Z$ is independent of $U_Y$ given $(W,V)$, then $Y(a)$ is
independent of $Z$ given $(W,V)$, analogous to
Assumption~\ref{ass:original_A3a}. Since $Y = f_Y(W, A, Z, U_Y)$ is already
a function of $Z$, assuming independence of $U_Z$ and $U_Y$ given $(W, V)$ is a
weak assumption. Finally, given that we already assumed no direct path of $V$ to
$Y$ (the path being blocked by $Z$) and that $U_Z$ is independent of $U_Y$ given
$(W, V)$, Assumption~\ref{ass:restrict_UY} does not add a much stronger
restriction on $U$---in fact, $U_V$ and $U_Z$ can still be dependent.

\section{Identification Results}

\subsection{The Controlled Direct Effect}\label{sec:cde-ident}

We define the following causal target parameter of interest, corresponding to
the controlled direct effect (CDE), as
\begin{align*}
    \Psi^F_{\text{CDE}}(\Pt_{(U,X),0}) &= \E_{\Pt_{(U,X),0}}[Y(1,z) - Y(0,z)].
\end{align*}

\vspace{3mm}
\begin{lemma}\label{lemma:cde-ident}
  Suppose the identifying Assumptions~\ref{ass:posA}, \ref{ass:posZ},
  \ref{ass:no-V-to-Y} and~\ref{ass:equalE} hold. Then, it follows that
\begin{equation*}
    \Psi^F_{\text{CDE}}(\Pt_{(U,X),0}) = \E[\E[Y \mid W, A = 1, Z = z] -
                                         \E[Y \mid W, A = 0, Z = z]].
\end{equation*}
\end{lemma}
\begin{proof}
  Note that by Lemma \ref{lemma:equivE}, it follows that $\E[Y(a,z) \mid W,V]
  = \E[Y(a,z) \mid W]$. Thus, the proof is a direct consequence of
  applying Lemma~\ref{lemma:equivE}.
\end{proof}

\subsection{The Natural Direct Effect}\label{sec:nde-ident}

In the main text, we define the following causal target parameter
of interest:
\begin{equation*}
    \Psi^F_{\text{NDE}}(\Pt_{(U,X),0}) = \int_{\mathcal{W}} \int_{\mathcal{Z}}
    \E[Y(1, z) - Y(0, z) \mid W = w] g_Z(z \mid w, A=0) p_W(w) \ .
\end{equation*}
The identification results are provided in the following lemma as well as under
Theorem~\ref{thm:id-NDE} of the main text. We emphasize that
the statistical target has two interpretations. The first corresponds to an
average of controlled direct effects, averaged over all values of $Z \in
\mathcal{Z}$. The second interpretation consists of the contrast in the
counterfactual mean outcome under exposure $A$ set deterministically as $A = 1$
or $A = 0$, and with $Z$ drawn from the conditional distribution with density
$g_{Z}(z \mid W, A=0)$. The two interpretations differ in the conditional
distribution of the mediator used; for the NDE, $P(Z(0) < z \mid W)$ is used as
opposed to $P(Z < z \mid W, A = 0)$. Note, however, that $P(Z(0) < z \mid W)$ is
in fact equal to $\P_{\Pt_0}(Z < z \mid W, A=0)$ by the consistency assumption
for identification.

\vspace{3mm}
\begin{lemma}
  Suppose the identifying Assumptions~~\ref{ass:posA}, \ref{ass:posZ},
  \ref{ass:no-V-to-Y} and~\ref{ass:equalE} hold. Then, we have that
\begin{align*}
    \Psi^F_{\text{NDE}}(\Pt_{(U,X),0}) =
       \E[\E[\E[Y \mid W, A = 1, Z] - \E[Y \mid W, A = 0, Z] \mid W, A = 0]]
       \ .
\end{align*}
\end{lemma}
%
\begin{proof}
  Again, by Lemma~\ref{lemma:equivE}, we have $\E[Y(a,z) \mid W, V] =
  \E[Y(a,z) \mid W]$. This result, then, is a direct consequence of
  Lemmas~\ref{lemma:equivE} and~\ref{lemma:cde-ident}.
\end{proof}

\subsection{The Average Treatment Effect}\label{sec:ate-ident}


We now consider the case in which our causal target parameter is the average
treatment effect (ATE), defined in Equation~\eqref{eqn:ate} as
\begin{equation*}
     \Psi_{\text{ATE}}^F(\Pt_{(U,X),0}) = \E_{\Pt_{(U,X),0}}[Y(1) - Y(0)] \ ,
\end{equation*}
where $Y(a) = f_Y(W, a, Z(a), U_Y)$ and $Z(a) = f_Z(W, V, a, U_Z)$ are the
counterfactuals of $Y$ and $Z$, respectively, under a static intervention that
sets $A=a$. In addition to all prior causal assumptions, by also assuming the
following (namely, Assumptions~\ref{ass:restrict_AZ} and~\ref{ass:cond_Z}), we
obtain an identification result for the counterfactual mean $E_{U,X}[Y(a)]$.

\vspace{2mm}
\begin{assumptionb}[\textit{Complete-data treatment-mediator
conditional randomization}]\label{ass:restrict_AZ}
\begin{align*}
U_A \indep U_Z \mid (W, V).
\end{align*}
\end{assumptionb}

\vspace{2mm}
\begin{assumptionb}[\textit{Equivalence of conditional distributions of the
mediator}]\label{ass:cond_Z}
\begin{align*}
\P_{\Pt_{(U,X),0}}(Z(1) < z \mid W) = \P_{\Pt_{(U,X),0}}(Z(0) < z \mid W).
\end{align*}
\end{assumptionb}

\vspace{3mm}
\begin{lemma}
  Suppose assumptions~\ref{ass:equalE},~
  \ref{ass:restrict_AZ} and~\ref{ass:cond_Z} hold. It immediately follows that
\begin{equation*}
    \E_{\Pt_{(U,X),0}}[Y(1)-Y(0)] = \E[\E[Y \mid W, A = 1, Z] -
      \E[Y \mid W, A = 0, Z]] \ .
\end{equation*}
\end{lemma}


\begin{proof}
For $a \in \{0, 1\}$, $\P_{\Pt_{(U,X),0}}(Z(a) < z \mid W) = \P_{\Pt_0}(Z < z
\mid W, A = a)$ by consistency and $\P_{\Pt_0}(Z < z \mid W, A = a) =
\P_{\Pt_0}(Z < z \mid W)$ by~\ref{ass:cond_Z}.
It suffices to show the identification result for the counterfactual mean
$\E_{\Pt_{(U,X),0}}[Y(1)]$, as analogous logic applies to
$\E_{\Pt_{(U,X),0}}[Y(0)]$. Under the assumptions enumerated above, we have that
\begin{align*}
   \E[Y(1)] &= \E[\E[Y(1) \mid W, V]] \\
   &= \E[\E[Y(1) \mid W, V, A=1]]
      \text{ by Assumption~\ref{ass:restrict_UA}, part of ~\ref{ass:equalE}} \\
   &= \E[\E[Y \mid W, V, A=1]] \text{ by consistency} \\
   &= \E[\E[\E[Y \mid W,V, A = 1, Z] \mid W, V, A=1]] \text{ by the law of iterated expectations}\\
   &= \E[\E[\E[Y \mid W, A=1, Z] \mid W,V, A=1]]
      \text{ by Assumption~\ref{ass:equalE}} \\
   &= \E[\E[\E[Y \mid W, A=1, Z(1)] \mid W,V, A=1]]
      \text{ by consistency} \\
   &= \E[\E[\E[Y \mid W, A = 1, Z(1)] \mid W,V]]
      \text{ by Assumption~\ref{ass:restrict_AZ}} \\
   &= \E[\E[Y \mid W, A = 1, Z]]
      \text{ by Assumption~\ref{ass:cond_Z}}.
\end{align*}
\end{proof}
We have established identifiability results for the NDE parameter in
Theorem~\ref{thm:id-NDE}, but with a slight twist: the average of
controlled direct effects w.r.t.~$\P_{\Pt_{(U,X),0}}(Z(0) < z \mid W)$ is
replaced by an average of controlled direct effects
w.r.t.~$\P_{\Pt_0}(Z < z \mid W, A=0)$. By additionally requiring
Assumption~\ref{ass:restrict_AZ}, we obtain a new estimand:
\begin{equation*}
  \E[\E[Y \mid W, A = 1, Z] - \E[Y \mid W, A = 0, Z]].
\end{equation*}
In words, under the assumption that $A$ is independent of $Z$, given baseline
covariates, our estimand for the NDE is equal to the estimand for the ATE. While
the derived estimand is identical to the ATE estimand, we cannot interpret it as
an ATE without Assumption~\ref{ass:cond_Z}. Effectively,
Assumption~\ref{ass:cond_Z} requires that any causal effect that $A$ may have on
$Z$ be statistically indistinguishable given $W$. Since $A$ is not randomized
given $W$, we do not have that $\P_{\Pt_{(U,X),0}}(Z(1) < z \mid W) =
\P_{\Pt_{(U,X),0}}(Z(0) < z \mid W)$ simply by assuming $A$ is independent of
$Z$ given $W$. Consequently, if we are willing to assume that
$\P_{\Pt_{(U,X),0}}(Z(1) < z \mid W) = \P_{\Pt_{(U,X),0}}(Z(0) < z \mid W)$,
then we can use the ATE estimand and also interpret the causal parameter as both
the ATE and the NDE. Intuitively, the natural indirect effect is zero---hence,
the ATE equals the NDE, making both interpretations appropriate.


\section{Statistical Inference} \label{sec:tmle}

We summarize the TMLE procedure derived by \citet{zheng2012}. We split the terms
of the efficient influence function in Equation~\eqref{eq:eif} of the main paper into orthogonal components as follows:
\begin{equation*}
  \begin{split}
    D^\star(O, \Pt)
    & = \color{red} \left(\frac{\I(A=1)}{g_A(W)}\frac{g_Z(Z \mid W,0)}
      {g_Z(Z \mid W,1)} - \frac{\I(A=0)}{1 - g_A(W)} \right)
      \left(Y - \overline{Q}_Y(W, A, Z) \right) \\
    & \qquad + \color{teal} \frac{\I(A=0)}{1 - g_A(W)}
    \left(\overline{Q}_Y(W,1,Z) - \overline{Q}_Y(W,0,Z) -
      \mathbb{E}_{\Pt}\left[\overline{Q}_Y(W, 1, Z) -
      \overline{Q}_Y(W, 0, Z)\big| W, A=0 \right]
    \right)\\
    & \qquad + \color{blue} \mathbb{E}_{\Pt}\left[\overline{Q}_Y(W, 1, Z) -
      \overline{Q}_Y(W,0,Z)\big| W, A=0 \right] \color{black} - \Psi(\Pt) \\
    & = \color{red} D^\star_Y(O, \Pt) \color{black} +
    \color{teal} D^\star_Z(O, \Pt) \color{black} +
    \color{blue} D^\star_W(O, \Pt) \color{black} - \Psi(\Pt) \ .
  \end{split}
\end{equation*}
Let the outcome $Y$ be a binary random variable or a continuous random variable
bounded, without loss of generality, between $(0, 1)$ almost surely (a.s.). We
consider the negative log-likelihood loss function for selecting
$\overline{Q}_Y$. That is,
\begin{equation*}
  L_Y(\overline{Q}_Y)(O) = -\log\left(
    \overline{Q}_Y(W, A, Z)^{Y}\left(1-\overline{Q}_Y(W, A, Z)\right)^{1-Y}
  \right),
\end{equation*}
which leads to the following parametric working submodel~\citep{vdl2011targeted}
for $\overline{Q}_Y$:
\begin{equation*}
  \overline{Q}_Y(\epsilon_1) \equiv \text{expit}\left(
    \text{logit}\left(\overline{Q}_Y\right) +
      \epsilon_1 \left(
        \frac{\I(A=1)}{g_A(W)} \frac{g_Z(Z \mid W,0)}{g_Z(Z \mid W,1)} -
        \frac{\I(A=0)}{1-g_A(W)}
      \right)
  \right)
\end{equation*}
such that
\begin{equation}\label{eq:loss-sol-cond-out}
  \frac{d}{d\epsilon_1} L_Y(\overline{Q}_Y(\epsilon_1))\big|_{\epsilon_1=0}
  = D^\star_Y(O, P).
\end{equation}

For any bounded $\overline{Q}_Y$, it holds that $\overline{Q}_Y(W, Z) =
\overline{Q}_Y(W, 1, Z) - \overline{Q}_Y(W, 0, Z)$ is also bounded. Let
$\overline{Q}_Y(W, Z)$ be bounded between $(0, 1)$ a.s.~without loss of
generality, and define $\psi_Z(W) \coloneqq \mathbb{E}_{P}[\overline{Q}_Y(W, 1,
Z) - \overline{Q}_Y(W, 0, Z)\big| W, A=0]$. Then, under the loss function given
below,
\begin{equation*}
  L_Z(\psi_Z)(O) = -\I(A=0)\log\left(
    \left(\psi_Z(W)\right)^{\overline{Q}_Y(W, Z)}
    \left(1 - \psi_Z(W)\right)^{(1 - \overline{Q}_Y(W, Z))}
  \right),
\end{equation*}
we obtain another logistic working submodel:
\begin{equation*}
  \psi_Z(\epsilon_2) \equiv \text{expit}\left(
    \text{logit}\left(\psi_Z\right) +
    \epsilon_2 \frac{1}{1-g_A(W)}
  \right).
\end{equation*}
Similarly to Equation~\ref{eq:loss-sol-cond-out}, we find that
\begin{equation}\label{eq:loss-sol-atc}
  \frac{d}{d\epsilon_2} L_Z(\psi_Z(\epsilon_2))\big|_{\epsilon_2=0}
  = D^\star_Z(O, \Pt).
\end{equation}

We can now describe the estimation procedure under the statistical model
described in Section~\ref{statistical_model} of the main text, with the added assumption that $Y$ is a bounded random variable between
$(0, 1)$. If $Y$ is not bounded between $(0, 1)$, then it suffices to transform
it to the unit interval by rescaling using $\{\min_{\Pt_n}(Y),
\max_{\Pt_n}(Y)\}$. Now, let $g_{A,n}$, $\overline{Q}_{Y,n}$ and $g_{Z,n}$ be
initial estimators of $g_{A,0}$, $\overline{Q}_{Y,0}$ and $g_{Z,0}$,
respectively, and let
\begin{equation*}
  \epsilon_{1,n} = \argmin_{\epsilon \in \R}
    \Pt_n L_Y\left(\overline{Q}_{Y,n}(\epsilon)\right)
\end{equation*}
such that $\overline{Q}_{Y,n}^\star \equiv \overline{Q}_{Y,n}(\epsilon_{1, n})$.
The univariate logistic regression model is used to tilt the initial estimates
such that $\Pt_n D_Y^\star(O, g_{A,n}, \overline{Q}_{Y,n}^\star, g_{Z,n})
\approx 0$~\citep[Chapter 5]{vdl2011targeted}. Denoting the estimator of
$\psi_Z(\Pt_0)$ by $\psi_{Z,n}$, we let
\begin{equation*}
  \epsilon_{2,n} = \argmin_{\epsilon \in \R}
    \Pt_n L_Z\left(\psi_Z(\epsilon)\right)
\end{equation*}
such that $\psi_{Z,n}^\star \equiv \psi_{Z,n}(\epsilon_{2,n})$. By the same
property of the loss function invoked previously, we obtain $\Pt_n D_Z^\star(O,
g_{A,n}, \overline{Q}_{Y,n}^\star, g_{Z,n}^\star) \approx 0$. The nonparametric
estimator of $p_{W,0}$, $p_{W,n}$ does not require the use of a targeting
procedure since it solves the relevant score equation already. That is,
$\Pt_n D_W^\star(O, \overline{Q}_{Y,n}^\star, \psi_Z^\star, p_{W,n}) = 0$.

Plugging in the targeted estimators of the nuisance parameters into the
substitution mapping provides the TMLE:
\begin{equation*}
  \Psi_\text{NDE}(\Pt_n^\star) = \frac{1}{n}\sum_{i=1}^n
  \psi_{Z,n}^\star(\overline{Q}_{Y,n}^\star)(W_i).
\end{equation*}
It follows from the above that plug-in bias term in
Equation~\eqref{eq:von-mises} of the main text converges to zero in probability.

\section{Efficient Inference Under Two-Phase Sampling
Designs}\label{sec:two-phase}

\subsection{The Efficient Influence Function}

In some applications, such as the immune correlates analysis of the CoVPN 3008
study considered in Section~\ref{sec:applic-covpn3008} of the main text, the mediator is measured in a subset of the sample,
which is selected by way of an \textit{outcome-dependent} sampling procedure
(n.b., such sampling techniques are generally termed ``two-phase'' designs). The
observed data unit is then represented by $O = (W, A, R, RZ, Y) \sim \Pt_0$.
Here, $W$, $A$, $Z$ and $Y$ have the same definitions as in
Section~\ref{statistical_model}, while the additional random variable, $R$,
is an indicator of $Z$ being measured. When $Z$ is not measured---that is, when
$R=0$---we assume $Z=0$, though this does not have any influence on downstream
inferential developments. Finally, parameters of $\Pt_{(U,X),0}$ are
identifiable using parameters of this modified $\Pt_0$ when the conditions of
Section~\ref{sec:identification} are satisfied and when the conditional
probability of measuring $Z$, denoted by $g_{R}(W,A,Y) = \P_{\Pt_0}[R=1 \mid
W,A,Y]$, meets the following positivity constraint: $\epsilon_R < g_{R}(W,A,Y)$
a.s.~for $\epsilon_R > 0$; note that this condition generally holds \textit{by
design} since the sampling procedure used to measure $Z$ is controlled by the
experimenter and should be specified explicitly at the study's design stage.

We again rely on the EIF to derive regular and asymptotically linear estimators.
A general mapping from the EIF of a pathwise differentiable target parameter
into the EIF of the same parameter under a two-phase sampling design was
provided by~\citet{vdl2011targeted}, with the aim of constructing a TML
estimator suitable for data produced under such designs; however, these authors
stopped short of interrogating or implementing their proposal. Later,
\citet{hejazi2020efficient} built on this result by adapting it to evaluating
the counterfactual mean of a modified treatment
policy~\citep{haneuse2013estimation, young2014identification,
diaz2018stochastic} subject to two-phase sampling of the exposure; these authors
proposed both a TML estimator and a one-step estimator and studied the
asymptotic efficiency properties of the resultant estimators when $g_R(W,A,Y)$
is flexibly estimated (and in a rate-optimal way). Somewhat anachronistically,
an earlier version of a similar result for asymptotically efficient estimation
under two-phase sampling was provided
by~\citet[see Lemma 6.1]{robins1994estimation} but seems to have escaped notice
in the two-phase sampling literature. We apply this procedure to the NDE, obtaining the following EIF:
\begin{equation*}
  D^F(O, \Pt) = \frac{\mathbb{I}(R=1)}{g_{R}(W,A,Y)} D^{\star}_R(O, \Pt)
  - \frac{\E_{\Pt}\left[D^{\star}_R(O, \Pt) \mid R=1,W,A,Y \right]}
  {g_{R}(W,A,Y)}(R - g_{R}(W,A,Y)) \ .
\end{equation*}
In the above equation, $D^F$ is the EIF of the population-level (or
``full-data'') parameter $\Psi(a, a^{\star}; \Pt_0)$ while $D^{\star}_R$ is the
EIF of the same target parameter at the distribution induced by the two-phase
sampling process (that is, in the phase-two sample $(W, A, Z, Y)$, for whom
$R = 1$), which may be expressed as $P \in \M_R = \{\Pt_{\Pt_0,g_{0,R}}: \Pt_0
\in \M, g_{0,R}\}$, where $g_{0,R}$ is the underlying two-phase sampling
mechanism. The EIF of the NDE or RR NDE can then be obtained by application of the
functional delta method.


\subsection{Estimation Algorithms}

The one-step estimator and TMLE of the NDE and RR NDE derived using the above
formula require a thoughtful (and involved!) estimation strategy. We present
cross-fitted algorithms for constructing both estimators of $\Psi(a, a^{\star};
P_{0})$ under a two-phase sampling design, assuming $g_{R}$ to be known (as in
the case-cohort sampling used in the CoVPN vaccine efficacy
trials~\citep{gilbert2021covpn}) and that the outcome is bounded between $0$ and
$1$ without loss of generality. Again, we highlight that estimators of the NDE
and RR NDE can be constructed from these estimators using the functional delta
method.

We note that \citet{benkeser2021inference} outline conditions under which these
estimators are multiply robust, asymptotically linear, and asymptotically
efficient. We invite interested readers to review their work for more
information about these estimators' asymptotic behavior.

\subsubsection{One-Step Estimator}

\begin{enumerate}
  \item Split the data into $K$ folds of approximately equal size, indexing them
        by $k=1, \ldots, K$.
  \item Combine folds $2$ through $K$ to form a training dataset, and let fold
        $1$ be the validation set. Let $P_{n}^{K-1}$ denote the
        empirical distribution of the training data, and $\Pt_{n}^{1}$ the
        empirical distribution of the validation data. Then,
    \begin{enumerate}
      \item Estimate $g_{A}(W, a)$ and $g_{A}(W, a^{\star})$, the conditional
            treatment assignment mechanisms of $A = \{a, a^{\star}\}$,
            respectively, conditioning on confounders, using all observations in
            the training dataset. This is achieved by regressing the confounders
            on the treatment indicator of $a$ and $a^{\star}$, respectively.
            Denote the resultant estimates by $g_{A,n}(W, a)$ and $g_{A,n}(W,
            a^\star)$.
      \item Estimate $g_{A}(W,a,Z)$ and $g_{A}(W,a^{\star},Z)$, the conditional
            treatment assignment mechanisms of $a$ and $a^{\star}$,
            respectively, conditioning on confounders and mediators, using the
            subset of training observations with measured mediators (for whom $R
            = 1$) and using two-phase sampling weights (assumed known). This is
            done by regressing the confounders and mediators on the treatment
            indicators $a$ and $a^\star$, respectively, \textit{using an
            estimation procedure capable of incorporating weights} (i.e., by
            re-weighting an appropriate loss function). We stress that
            regression procedures incapable of incorporating the two-phase
            sampling weights are incompatible with this step. Denote the
            resultant estimates by $g_{A, n}(W,a,Z)$ and $g_{A,
            n}(W,a^{\star},Z)$.
      \item Estimate $\overline{Q}_{Y}(W,A,Z)$, the conditional expected outcome,
            conditioning on the treatment, confounders, and mediators, using the
            subset of training observations with measured mediators (for whom $R
            = 1$) and using the two-phase sampling weights (assumed known). This
            is done by regressing the treatment indicator, confounders, and
            mediators on the outcome. As with the preceding step, the estimation
            procedure must be capable of incorporating the two-phase sampling
            weights. Denote the estimate by $\overline{Q}_{n}(W,A,Z)$.
      \item Estimate $v_{Y}(W, A, a) = \E[\overline{Q}_{n}(W, a, Z) \mid W, A]$,
            the conditional expected pseudo-outcome under treatment level $a$,
            conditioning on the treatment and confounders, by regressing the
            treatment indicators and confounders on $\overline{Q}_{n}(W, a, Z)$
            using the subset of training data with measured mediators (for whom
            $R = 1$). Note that this regression step \textit{should ignore} the
            two-phase sampling weights, as these have been accounted for in
            preceding steps. Let the estimate be given by $v_{Y,n}(W, A, a)$.
      \item Compute the plug-in estimator
            $\psi_{n} = \E_{\Pt_{n}^{K-1}}[v_{Y,n}(W, a^{\star}, a)]$.
      \item Compute the empirical EIF, denoted by $D_{n}^{\star}$,
            using the training dataset's observations who have measured
            confounders and the nuisance parameter estimates. For any given
            observation $O$, their corresponding entry in $D_{n}^{\star}$ is
            given by
            \begin{align*}
              D^{\star}_{a, a^{\star}}(O, \Pt_{n}^{k-1})
              & = \frac{\I(A=a)}{\cancel{g_{A,n}(W, a)}}
              \frac{\cancel{g_{A,n}(W, a)}}{g_{A,n}(W, a^\star)}
              \frac{g_{A,n}(W,a^\star,Z)}{g_{A,n}(W,a,Z)}
                \left(Y-\overline{Q}_{Y,n}(W,a,Z)\right) \\
              & \quad + \frac{\I(A = a^\star)}{g_{A,n}(W,a^\star)}
                \left(\overline{Q}_{Y,n}(W, a, Z) - v_{Y,n}
                (W,a^{\star}, a)\right) \\
              & \quad + v_{Y,n}(W,a^{\star}, a) - \psi_{n} \; .
            \end{align*}
      \item Estimate the conditional EIF, $\E_{P_{0}}[D^*_{a, a^{\star}}(O,
            \Pt_0) \mid W,A,Y]$, by regressing the confounders, treatment
            indicator, mediators, and outcome of the training data with measured
            mediators (for whom $R = 1$) on $D_{n}^{\star}$. Denote the
            estimate by $d_{n}(W,A,Y)$. The empirical full-data EIF is then
            given by
            \begin{equation*}
              D^{F}_{n}(O, \Pt_{n}^{K-1}) = \frac{\I(R=1)}{g_{R}(W,A,Y)}
              D^{\star}_{a, a^{\star}}(O, \Pt_{n}^{K-1})
              - \left(\frac{d_{n}(W,A,Y)}{g_{R}(W,A,Y)}\right)
              (R - g_{R}(W,A,Y)) \ .
            \end{equation*}
      \item The full-data EIF can now be predicted for all observations in the
            validation sample using the EIF $D^{F}_{n}(O, \Pt_{n}^{K-1})$
            estimated in the training sample. This EIF is $D^{F}_{n, 1}$.
      \item The un-centered full-data EIF is given by
            $\widetilde{D}^{F}_{n, 1} = D^{F}_{n, 1} + \psi_{n}$.
    \end{enumerate}
  \item Repeat the previous step $K-1$ more times, changing the training and
        validation samples each time (until the index set $k = 1, \ldots, K$
        has been exhausted), and concatenate the $K$ vectors of
        un-centered predicted full-data EIFs
        $\widetilde{D}^{F}_{n} = (\widetilde{D}^{F}_{n, 1},
        \widetilde{D}^{F}_{n, 2}, \ldots, \widetilde{D}^{F}_{n, K})$. This is a
        vector of $n$ scalars, one for every observation in the dataset.
  \item The cross-fitted one-step estimator is then given by
       $\psi_{n}^{\text{CV-OS}} = \E_{\Pt_n}[\widetilde{D}^{F}_{n}]$. Its
       variance is $\mathbb{V}_{\Pt_n}[\widetilde{D}^{F}_{n}] / n$.
\end{enumerate}

\subsubsection{Targeted Maximum Likelihood Estimator}

\begin{enumerate}
  \item Split the data into $K$ folds of approximately equal size, indexing them
        by $k=1, \ldots, K$.
  \item Combine folds $2$ through $K$ to form a validation dataset, and let fold
        $1$ be the validation set. Let $P_{n}^{K-1}$ denote the
        empirical distribution of the training data, and $P_{n}^{1}$ the
        empirical distribution of the validation data. Then,
    \begin{enumerate}
      \item Estimate $g_{A}(W, a)$ and $g_{A}(W, a^{\star})$, the conditional
            treatment assignment mechanisms of $A = \{a, a^{\star}\}$,
            respectively, conditioning on confounders, using all observations in
            the training dataset. This is achieved by regressing the confounders
            on the treatment indicator of $a$ and $a^{\star}$, respectively.
            Denote the resultant estimates by $g_{A,n}(W, a)$ and $g_{A,n}(W,
            a^\star)$.
      \item Estimate $g_{A}(W,a,Z)$ and $g_{A}(W,a^{\star},Z)$, the conditional
            treatment assignment mechanisms of $a$ and $a^{\star}$,
            respectively, conditioning on confounders and mediators, using the
            subset of training observations with measured mediators (for whom $R
            = 1$) and using two-phase sampling weights (assumed known). This is
            done by regressing the confounders and mediators on the treatment
            indicators $a$ and $a^\star$, respectively, \textit{using an
            estimation procedure capable of incorporating weights} (i.e., by
            re-weighting an appropriate loss function). We stress that
            regression procedures incapable of incorporating the two-phase
            sampling weights are incompatible with this step. Denote the
            resultant estimates by $g_{A, n}(W,a,Z)$ and $g_{A,
            n}(W,a^{\star},Z)$.
      \item Estimate $\overline{Q}_{Y}(W,A,Z)$, the conditional expected outcome,
            conditioning on the treatment, confounders, and mediators, using the
            subset of training observations with measured mediators (for whom $R
            = 1$) and using the two-phase sampling weights (assumed known). This
            is done by regressing the treatment indicator, confounders, and
            mediators on the outcome. As with the preceding step, the estimation
            procedure must be capable of incorporating the two-phase sampling
            weights. Denote the estimate by $\overline{Q}_{n}(W,A,Z)$.
      \item Estimate $v_{Y}(W, A, a) = \E[\overline{Q}_{n}(W, a, Z) \mid W, A]$,
            the conditional expected pseudo-outcome under treatment level $a$,
            conditioning on the treatment and confounders, by regressing the
            treatment indicators and confounders on $\overline{Q}_{n}(W, a, Z)$
            using the subset of training data with measured mediators (for whom
            $R = 1$). Note that this regression step \textit{should ignore} the
            two-phase sampling weights, as these have been accounted for in
            preceding steps. Let the estimate be given by $v_{Y,n}(W, A, a)$.
      \item Compute the plug-in estimator
            $\psi_{n} = \E_{\Pt_{n}^{K-1}}[v_{Y,n}(W, a^{\star}, a)]$.
      \item Compute the empirical EIF, denoted by $D_{n}^{\star}$,
            using the training dataset's observations who have measured
            confounders and the nuisance parameter estimates. For any given
            observation $O$, their corresponding entry in $D_{n}^{\star}$ is
            given by
            \begin{align*}
              D^{\star}_{a, a^{\star}}(O, \Pt_{n}^{k-1})
              & = \frac{\I(A=a)}{\cancel{g_{A,n}(W, a)}}
              \frac{\cancel{g_{A,n}(W, a)}}{g_{A,n}(W, a^\star)}
              \frac{g_{A,n}(W,a^\star,Z)}{g_{A,n}(W,a,Z)}
                \left(Y-\overline{Q}_{Y,n}(W,a,Z)\right) \\
              & \quad + \frac{\I(A = a^\star)}{g_{A,n}(W,a^\star)}
                \left(\overline{Q}_{Y,n}(W, a, Z) - v_{Y,n}(W,a^{\star}, a)\right) \\
              & \quad + v_{Y,n}(W,a^{\star}, a) - \psi_{n} \; .
            \end{align*}
      \item Estimate the conditional EIF, $\E_{\Pt_0}[D^*_{a, a^{\star}}(O,
            \Pt_0) \mid W,A,Y]$, by regressing the confounders, treatment
            indicator, mediators, and outcome of the training data with measured
            mediators (for whom $R = 1$) on $D_{n}^{\star}$. Denote the
            estimate by $d_{n}(W,A,Y)$. The empirical full-data EIF is then
            given by
            \begin{equation*}
              D^{F}_{n}(O, \Pt_n^{K-1}) = \frac{\mathbb{I}(R=1)}{g_{R}(W,A,Y)}
              D^{\star}_{a, a^{\star}}(O, \Pt_n^{K-1})
              - \left(\frac{d_{n}(W,A,Y)}{g_{R}(W,A,Y)}\right)
              (R - g_{R}(W,A,Y)) \ .
            \end{equation*}
      \item The full-data EIF can now be predicted for all observations in the
            validation sample using the EIF $D^{F}_{n}(O, \Pt_{n}^{K-1})$
            estimated in the training sample. This EIF is $D^{F}_{n, 1}$.

      \item Compute the relevant score equations:
            \begin{enumerate}
              \item The two-phase sampling weights score is given by:
                    \begin{equation*}
                      \E_{\Pt_n^1}
                      \left[\frac{d_{n}(W,A,Y)}{g_{R}(W,A,Y)}
                        (\mathbb{I}(R = 1) - g_{R}(W,A,Y)) \right] \ .
                    \end{equation*}

              \item The conditional expected outcome score is given by:
                    \begin{equation*}
                      \E_{\Pt_{n}^{1}}\left[
                      \frac{\mathbb{I}(R = 1)}{g_{R}(W,A,Y)}
                      \frac{\mathbb{I}(A = a)}{\cancel{g_{A,n}(W,a)}}
                      \frac{\cancel{g_{A,n}(W,a)}}{g_{A,n}(W,a^{\star})}
                      \frac{g_{A,n}(W,a^{\star},Z)}{g_{A,n}(W,a,Z)}
                      (Y - \overline{Q}_{Y,n}(W,a,Z)) \right] \ .
                    \end{equation*}
            \end{enumerate}

      \item Repeat the following steps until the score equations are
            approximately equal to zero:
            \begin{enumerate}
              \item Optimizing the working submodel
                    $g_{R,n}(\epsilon) = \text{expit}(\text{logit}(1/g_{R}) +
                    \epsilon d_{n} / g_{R})$ with respect to $\epsilon$ using
                    a logistic regression model of the two-phase sampling
                    indicator, we obtain the updated $g_{R,n}^{\star}(W, A, Y)
                    = g_{R}(\epsilon_{n})$. Here, $\epsilon_{n}$ is the
                    estimated coefficient obtained by fitting the logistic
                    regression model. $g_{R}(W,A,Y)$ is then replaced by
                    $g_{R,n}^{\star}(W,A,Y)$ in the following iteration.
              \item Using a similar working submodel
                    $\overline{Q}_{Y,n}(\epsilon) = \text{expit}(\text{logit}
                    (\overline{Q}_{Y,n}(W, A)) + \epsilon (g_{A,n}(W,a)g_{A,n}
                    (W,a^{\star},Z))/ (g_{A,n}(W,a^{\star})g_{A,n}(W,a,Z))$
                    of $Y$, we obtain $\epsilon_{n}$ by fitting a weighted
                    logistic regression. The weights are given by
                    $1/(g_{A,n}(W, a)g_{R,n}^{\star}(W,A,Y))$. An updated
                    conditional expected outcome estimate is then given by
                    $\overline{Q}_{Y,n}^{\star}(W,A,Z) =
                    \overline{Q}_{Y,n}(\epsilon_{n})$.
                    $\overline{Q}_{Y,n}(W,A,Z)$ is replaced by
                    $\overline{Q}_{Y,n}^{\star}(W,A,Z)$ in the next iteration.
              \item Update the scores given in the previous step, replacing
                    $g_{R}$ and $\overline{Q}_{Y,n}$ by $g^{\star}_{R,n}$ and
                    $\overline{Q}_{Y,n}^{\star}$, respectively.
            \end{enumerate}
      \item $v_{Y,n}(W, A, a)$ is updated by fitting the following working
            submodel,
            $v_{Y,n}(\epsilon) = \text{expit}(\text{logit}(v_{Y,n}(W,A,a)))$,
            using a weighted logistic regression whose outcome is
            $\overline{Q}_{Y,n}(W,A,Z)$. The weights are equal to
            $I(A=a^\star)/(g_{A,n}(W, a^{\star})g_{R,n}^{\star}(W,A,Y))$. The
            updated $v_{n}(W, A, a)$ is given by
            $v_{Y,n}^{\star}(W, A, a) = v_{Y,n}(\epsilon_{n})$, where
            $\epsilon_{n}$ is the estimated coefficient obtained by fitting the
            logistic regression to the submodel.
      \item Compute the TML estimate over the validation fold:
            $\psi_{n,1}^{\text{TML}} = \E_{P_{n}^{1}}
            [v_{Y,n}^{\star}(W, a^{\star}, a)]$.

    \end{enumerate}
  \item Repeat the previous step $K-1$ more times, changing the training and
        validation samples each time (until the index set $k = 1, \ldots, K$
        has been exhausted), and concatenate the $K$ vectors of predicted
        full-data EIFs $D^{F}_{n} = (D^{F}_{n, 1}, D^{F}_{n, 2}, \ldots,
        D^{F}_{n, K})$. This is a vector of $n$ scalars, one for every
        observation in the dataset.
  \item The cross-validated TML estimator is then given by
     $\psi_{n}^{\text{CV-TML}} = \sum_{k=1}^{K} \psi_{n,k}^{\text{TML}} / K$.
        Its variance is $\mathbb{V}_{\Pt_{n}}[D^{F}_{n}] / n$.
\end{enumerate}

\newpage
\FloatBarrier

\bibliography{refs}

\end{document}